\documentclass[letterpaper]{scrartcl}
\pdfoutput=1
\usepackage[total={16.2cm, 23.5cm}]{geometry}
\usepackage{microtype}
\usepackage{amsmath,amsthm,amssymb}
\usepackage{tikz}
\usetikzlibrary{patterns,decorations.pathreplacing}
\usepackage[pagebackref]{hyperref}
\hypersetup{
pdfencoding=auto, psdextra,
colorlinks=true,citecolor=green!50!black,linkcolor=red!60!black}

\usepackage[sort&compress,nameinlink]{cleveref}
\usepackage{nicefrac}
\usepackage{natbib}
\usepackage{paralist}
\usepackage{xcolor}
\usepackage{colortbl}
\usepackage{subcaption}
\usepackage{booktabs}
\usepackage[linecolor=green!70!black, backgroundcolor=green!10, bordercolor=black, textsize=small]{todonotes}
\definecolor{winered}{rgb}{0.6,0.1,0.1}

\renewcommand*{\le}{\leqslant}
\renewcommand*{\leq}{\leqslant}
\renewcommand*{\ge}{\geqslant}
\renewcommand*{\geq}{\geqslant}
\renewcommand{\epsilon}{\varepsilon}

\newcommand{\myemph}[1]{{\color{winered}\emph{#1}}}
\newcommand{\naturals}{{{\mathbb{N}}}}

\theoremstyle{definition}
\newtheorem{definition}{Definition}
\newtheorem{example}{Example}
\newtheorem{remark}{Remark}
\theoremstyle{theorem}
\newcommand{\score}{{{\mathrm{sc}}}}
\newtheorem{theorem}{Theorem}

\newtheorem{proposition}{Proposition}

\makeatletter
\newtheorem*{rep@theorem}{\rep@title}
\newcommand{\newreptheorem}[2]{\newenvironment{rep#1}[1]{\def\rep@title{#2 \ref{##1}}\begin{rep@theorem}}{\end{rep@theorem}}}
\makeatother
\newreptheorem{theorem}{Theorem}
\newreptheorem{proposition}{Proposition}
\newreptheorem{lemma}{Lemma}
\newreptheorem{corollary}{Corollary}

\newcommand{\argmax}{{\mathrm{argmax}}}
\newcommand{\argmin}{{\mathrm{argmin}}}

\begin{document}
	\title{Proportionality and the Limits of Welfarism}
	\author{Dominik Peters \\ {\large Harvard University} \\ {\large dpeters@seas.harvard.edu}
	 \and Piotr Skowron \\ {\large University of Warsaw} \\ {\large p.skowron@mimuw.edu.pl}}
	\date{\normalsize First Published: November 2019. Last Update: October 2022.}
	\maketitle

\begin{abstract}
We study two influential voting rules proposed in the 1890s by Phragm\'en and Thiele, which elect a committee or parliament of $k$ candidates which proportionally represents the voters. Voters provide their preferences by approving an arbitrary number of candidates. Previous work has proposed proportionality axioms satisfied by Thiele's rule (now known as Proportional Approval Voting, PAV) but not by Phragm\'en's rule. By proposing two new proportionality axioms (laminar proportionality and priceability) satisfied by Phragm\'en but not Thiele, we show that the two rules achieve two distinct forms of proportional representation. Phragm\'en's rule ensures that all voters have a similar amount of influence on the committee, and Thiele's rule ensures a fair utility distribution.

Thiele's rule is a welfarist voting rule (one that maximizes a function of voter utilities). We show that no welfarist rule can satisfy our new axioms, and we prove that no such rule can satisfy the core. Conversely, some welfarist fairness properties cannot be guaranteed by Phragm\'en-type rules. This formalizes the difference between the two types of proportionality.  We then introduce an attractive committee rule which satisfies a property intermediate between the core and extended justified representation (EJR). It satisfies laminar proportionality, priceability, and is computable in polynomial time. We show that our new rule provides a logarithmic approximation to the core. On the other hand, PAV provides a factor-2 approximation to the core, and this factor is optimal for rules that are fair in the sense of the Pigou--Dalton principle.
\end{abstract}

\section{Introduction}

In the mid-1890s, two Nordic mathematicians were engaged in a somewhat heated academic exchange about voting. Sweden was soon to introduce universal suffrage, and the conservatives expected to become a minority and were thus pushing for electoral reform: what was needed was an election system that would guarantee \myemph{proportional representation} in parliament. 
Edvard Phragm\'en proposed a voting rule in 1894, and Thorvald N. Thiele responded with an alternative rule in 1895. They each proved that their rule guarantees proportionality on well-behaved inputs. Thiele's rule is simpler to understand, and thus it was popular and soon applied. Phragm\'en thought this was premature, and in 1899 published some examples to show this his rule was better. Many years later, it is still not clear who came up with the better system.\footnote{This historical account follows \citet{Janson16arxiv}.}

Both rules are of the same type, and are based on what is now called \myemph{approval voting}: Each voter $i \in N$ gets a ballot listing all $m$ candidates, and is allowed to approve an arbitrary subset~$A_i$ of them. The voting rule then selects a winning \myemph{committee} (or parliament) of $k$ candidates. (For an overview of the literature on approval-based committee elections we refer to the recent survey by \citet{lac-sko:abc-survey}.)

After the initial papers, rigorous study of the rules was interrupted until a few years ago, when Thiele's method was rediscovered---today, it is commonly called \myemph{Proportional Approval Voting} (PAV)---and there has been intense interest in formalizing the idea of proportional representation. A one-dimensional hierarchy of axioms has been proposed (with ``extended justified representation'' (EJR) being the strongest axiom); PAV does very well on those axioms, and Phragm\'en's rule perhaps a bit less.

What has not been recognized is that Thiele's and Phragm\'en's rules achieve proportionality in two fundamentally distinct senses. These two philosophies cannot be captured by the existing hierarchy, and are formally incompatible. This paper attempts to clarify the difference. It shows that the type of proportionality that PAV implements can be viewed as a fair distribution of welfare, whereas the one implemented by Phragm\'en's rules as a fair distribution of voting power.

\subsection*{An example} 

We begin our discussion by considering an example where the two rules behave differently. We defer formal definitions until later, and only give an intuitive sense of the rules' behavior.
There are $n=6$ voters and $m = 15$ candidates running for $k = 12$ seats. The figure below indicates voters' approval sets, to be read in columns. For example, voter $v_1$ approves $\{c_1, c_2, c_3, c_4\}$.

\[
\begin{tikzpicture}
[yscale=0.43,xscale=0.78,voter/.style={anchor=south, yshift=-7pt}, select/.style={fill=blue!10}, c/.style={anchor=south, yshift=1.5pt, inner sep=0}]
	\draw[select] (0,0) rectangle (3,1);
	\draw[select] (0,1) rectangle (3,2);
	\draw[select] (0,2) rectangle (3,3);
	\draw[select] (0,3) rectangle (1,4);
	\draw[select] (1,3) rectangle (2,4);
	\draw[select] (2,3) rectangle (3,4);
	\node at (1.5,0.42) {$c_1$};
	\node at (1.5,1.42) {$c_2$};
	\node at (1.5,2.42) {$c_3$};
	\node at (0.5,3.42) {$c_4$};
	\node at (1.5,3.42) {$c_5$};
	\node at (2.5,3.42) {$c_6$};
	\foreach \x in {3,4,5}
		{
		\foreach \y in {0,1}
			{
			\draw[select] (\x,\y) rectangle (\x+1,\y+1);
			}
		\foreach \y in {2}
			{
			\draw (\x,\y) rectangle (\x+1,\y+1);
			}
		}
	\node at (3.5,0.42) {$c_{7}$};
	\node at (3.5,1.42) {$c_{8}$};
	\node at (3.5,2.42) {$c_{9}$};
	\node at (4.5,0.42) {$c_{10}$};
	\node at (4.5,1.42) {$c_{11}$};
	\node at (4.5,2.42) {$c_{12}$};
	\node at (5.5,0.42) {$c_{13}$};
	\node at (5.5,1.42) {$c_{14}$};
	\node at (5.5,2.42) {$c_{15}$};
	\foreach \i in {1,...,6}
		\node[voter] at (\i-0.5,-1) {$v_\i$};
		
	\node at (3, -2.5) {(a) Phragm\'en's rule};
\end{tikzpicture}
\qquad\quad
\begin{tikzpicture}
[yscale=0.43,xscale=0.78,voter/.style={anchor=south, yshift=-7pt}, select/.style={fill=blue!10}, c/.style={anchor=south, yshift=1.5pt, inner sep=0}]
	\draw[select] (0,0) rectangle (3,1);
	\draw[select] (0,1) rectangle (3,2);
	\draw[select] (0,2) rectangle (3,3);
	\draw (0,3) rectangle (1,4);
	\draw (1,3) rectangle (2,4);
	\draw (2,3) rectangle (3,4);
	\node at (1.5,0.42) {$c_1$};
	\node at (1.5,1.42) {$c_2$};
	\node at (1.5,2.42) {$c_3$};
	\node at (0.5,3.42) {$c_4$};
	\node at (1.5,3.42) {$c_5$};
	\node at (2.5,3.42) {$c_6$};
	\foreach \x in {3,4,5}
		{
		\foreach \y in {0,1,2}
			{
			\draw[select] (\x,\y) rectangle (\x+1,\y+1);
			}
		}
	\node at (3.5,0.42) {$c_{7}$};
	\node at (3.5,1.42) {$c_{8}$};
	\node at (3.5,2.42) {$c_{9}$};
	\node at (4.5,0.42) {$c_{10}$};
	\node at (4.5,1.42) {$c_{11}$};
	\node at (4.5,2.42) {$c_{12}$};
	\node at (5.5,0.42) {$c_{13}$};
	\node at (5.5,1.42) {$c_{14}$};
	\node at (5.5,2.42) {$c_{15}$};
	\foreach \i in {1,...,6}
		\node[voter] at (\i-0.5,-1) {$v_\i$};
		
	\node at (3, -2.5) {(b) Thiele's rule (PAV)};
\end{tikzpicture}
\]
The committees selected by the two rules are shaded in blue. Which is the better choice?

Phragm\'en's rule looks at the given profile and treats voters $v_1$, $v_2$, $v_3$ as a group, since their approved candidates are disjoint from others'. Since these voters form half the electorate, they get to decide half of the available seats, so all six candidates are elected. The rule treats the remaining voters as singleton groups, and thus assigns them each a sixth of the available seats.  We will later formalize this behavior in an axiom called \myemph{laminar proportionality}, which encodes a kind of ``procedural'' proportionality. The committee (a) can also be justified as an outcome that results if we assign each voter equal ``power'', which we will formalize later using virtual money that voters can spend to elect candidates; we will say that the committee (a) is \myemph{priceable}. Because the first three voters approve three common candidates, they can share the cost of electing those candidates, and still have a sufficient budget to afford $c_4$, $c_5$, and $c_6$.

In comparison, the committee (b) will not result if we give each voter equal power: the first three and the last three voters have the same overall budget, yet the former group gets only three candidate, and the latter nine. Thus, (b) is not priceable. From a related point of view, committee~(b) is \myemph{unstable} in a sense borrowed from cooperative game theory: the first half of the voters can be said to ``deserve'' to decide six seats, and they might deviate by proposing six candidates on their own. Thus committee (b) is not in the \myemph{core}.

Why, then, does Thiele's rule select (b)? The reason is that PAV does not aim for fair distribution of power, but rather for ``welfare proportionality''. Welfare in the committee context is based on the number of candidates in the committee that the voter approves. Looking at the welfare vectors induced by (a) and (b), Thiele's choice seems more compelling according to several criteria. First, committee (b) is better than (a) according to \citeauthor{kelly1997charging}'s \citeyearpar{kelly1997charging} \myemph{proportional fairness} criterion: if we moved from (b) to (a), we would make three voters worse off by 33\% and make three candidates better off by 25\%, which on average is a worsening. Second, committees (a) and (b) have the same utilitarian welfare (the sum of utilities is 18), but (b) spreads representation more equally. Finally, committee (a) violates the \myemph{Pigou--Dalton principle}: by removing candidate $c_6$ and adding $c_9$, we keep the welfare of all but two voters constant, but have reduced the inequality between $v_3$ and $v_4$. Since it allows such a transfer of utility, committee (a) seems anti-egalitarian.

Thus, we see that choice (a) follows the spirit of proportionality very closely, and it splits the committee seats in proportion to group sizes. Choice (b) places a greater emphasis on welfare, and aims for fairness in those terms.

\subsection*{Proportionality axioms}

\paragraph{Thiele's method.}
To decide which $k$ candidates are picked, PAV solves an optimization problem: it elects a committee $W$ of $k$ candidates maximizing the following quantity: 
\[ \sum_{i \in N} 1 + \frac12 + \frac13 + \cdots + \frac{1}{|W \cap A_i|}.  \]
Thus, the winning committee $W$ maximizes the sum of voter ``satisfaction'', where satisfaction is an increasing function in the number $|W \cap A_i|$ of committee members that the voter approves. Specifically, Thiele noticed that using the harmonic number guarantees that the winning committee is \emph{proportional}. By ``proportional'', he meant that in situations where each voter approves exactly the members of a single party, the committee is divided in proportion to party support. For example, suppose there are $k = 100$ seats, and 30\% of the voters approve candidates $\{a_1, \dots, a_{100}\}$, another 30\% approve $\{b_1, \dots, b_{100}\}$, and the remaining 40\% approve $\{c_1, \dots, c_{100}\}$. Then Thiele's method elects 30 candidates $a_j$, 30 candidates $b_j$, and 40 candidates $c_j$. In any such ``completely partisan'' preference profile, Thiele's method identifies the proportional solution, and one can show that the harmonic numbers are the unique choice in satisfaction function with this property~\citep{bri-las-sko:c:apportionment,lac-sko:t:approval-thiele}.\footnote{Thiele only considered profiles where for each party, the fraction of voters supporting it is a multiple of $1/k$, thus avoiding rounding issues. To uniquely characterize harmonic numbers, one needs to also specify what happens in cases rounding issues arise, for example by specifying that the rule should act according to the D'Hondt apportionment method.} 

Thiele's notion of proportionality is weak and only applies to cases where different votes are disjoint. We would also like to have proportionality guarantees when different approval sets may intersect. In this general case, defining proportionality is more subtle. Under one intuition, a ``cohesive'' group of voters (whose approval sets have a large intersection) should be represented in proportion to their size. Recent research has formalized this idea using notions like \emph{extended justified representation (EJR)}~\citep{justifiedRepresenattion}, which PAV satisfies. Further, PAV guarantees that cohesive groups obtain high average satisfaction~\citep{justifiedRepresenattion, AEHLSS18}.

We say that a rule is \emph{welfarist} if it decides the winning committee only based on welfare information. Formally, a welfarist rule selects the committees that maximize some function of the \emph{welfare vector} $(|W \cap A_i|)_{i \in N}$ of a committee $W$, which specifies the number of approved committee members for each voter. Clearly, PAV is a welfarist rule.  Now, it is obvious that notions like Pareto-optimality or the Pigou--Dalton principle can be captured by welfarist rules. In contrast, proportionality on first sight requires a look at the internal structure of reported preferences, if only to identify cohesive groups that need to be represented. But PAV shows that welfare information is sufficient to give strong proportionality guarantees, such as EJR. \emph{Our aim is to explore the limits of welfarism in capturing various forms of proportionality.}

\paragraph{Phragm\'en's objection.}
In one possible description of Phragm\'en's rule, imagine that each voter has a bank account, initially with no money in it. At a common and constant rate, each account is continuously topped up with money. As soon as there is a candidate $c$ whose approvers have a combined account balance summing to at least (by continuity, equal to) \$1, the candidate $c$ is elected and removed, and the balance of each approving voter is set to \$0. We then continue. Phragm\'en's rule returns the committee consisting of the first $k$ candidates elected this way. Phragm\'en called the virtual money ``voting power''. 

In an 1899 article, Phragm\'en points out some downsides of PAV. He considers an example (here slightly adapted to avoid ties) where 3000 voters approve $\{a, b_1, b_2, b_3, b_4\}$ and 1000 voters approve $\{a, c_1, c_2, c_3, c_4\}$. There are $k = 5$ seats to be filled. Intuitively, there are two parties: $B = \{b_1, b_2, b_3, b_4\}$, supported by 75\% of voters, and $C = \{c_1, c_2, c_3, c_4\}$, supported by 25\%. Voters are partisan and approve exactly one of the parties, except that there is a consensus candidate $a$ who is approved by all voters. Phragm\'en's rule elects a committee of the form $\{a, b_i, b_j, b_k, c_l\}$: one seat is filled with the consensus candidate, and the remaining 4 seats are filled by party members in proportion to the party's support. Intuitively, this is a proportional choice. However, the PAV score of $\{a, b_1, b_2, b_3, c_1\}$ is $3000 \cdot (1 + \frac12 + \frac13 + \frac14) + 1000 \cdot (1 + \frac12) = 7750$, while the PAV score of $\{a, b_1, b_2, b_3, b_4\}$ is $3000 \cdot (1 + \frac12 + \frac13 + \frac14 + \frac15) + 1000 \cdot 1 = 7850 > 7750$, so the latter is the committee chosen by PAV. 

Phragm\'en (1899) argues that PAV's behavior on the above example is undesirable,\footnote{Phragm\'en writes: ``As we can see, Thiele's method benefits the larger party at the expense of the smaller one. This obviously means that, when using Thiele's method, the smaller party could never come to an agreement wherein both parties vote for one or more candidates that are not aligned with either party. To me it has always seemed to be the most important, or in any case one of the most important, requirements to impose on a proportional election method, that it must not obstruct such agreements between parties. It therefore seems to me that the mentioned property of Thiele's method is a very serious flaw.'' In a sense, PAV removes incentives to compromise and rewards partisanship.}
 and indeed the committees returned by his rule in this scenario appear more intuitive.
On the other hand, Phragm\'en's rule fails the the EJR property mentioned above \citep{aaai/BrillFJL17-phragmen}. \emph{Our aim is to clarify the differences in the types of proportionality provided by the two rules.}

\subsection*{Summary of our main results}

\begin{table}[t]
	\centering
	\makebox[\textwidth][c]{
	\begin{tabular}{lccc}
		\toprule
		& Thiele's method (PAV) & Phragm\'en's method & Equal Shares \\
		\midrule
		laminar proportional & & $\checkmark$ & $\checkmark$ \\
		priceable & & $\checkmark$ & $\checkmark$ \\
		PJR & $\checkmark$ & $\checkmark$ & $\checkmark$ \\
		EJR & $\checkmark$ & & $\checkmark$ \\
		core with constrained deviations\!\!\! & & & $\checkmark$ \\
		core approximation & 2-approx. & ? & $O(\log k)$-approx. \\
		\midrule
		welfarist & $\checkmark$ & & \\
		Pareto-optimal & $\checkmark$ & & \\
		Pigou--Dalton & $\checkmark$ & & \\
		\midrule
		computation & NP-complete & polynomial time & polynomial time \\
		\bottomrule
	\end{tabular}
	}
	\caption{The rules we consider and properties that they satisfy.}\label{tab:properties_summary}
\end{table}

Based on Phragm\'en's arguments we identify some general classes of situations in which PAV seems to fail to provide proportionality. These failures inspire the definition of two new proportionality axioms, \emph{laminar proportionality} and \emph{priceability}, that illuminate the difference between the two rules: these axioms are satisfied by Phragm\'en's rule, but failed by Thiele's rule. We show that PAV fails these notions because they are intrinsically not welfarist: we prove that no welfarist rule can satisfy either of our two new axioms. 

This technique also helps explain why PAV fails to select committees in the \emph{core}~\citep{justifiedRepresenattion, FMS18}: A committee is said to be in the \emph{core} if no coalition of voters blocks it. Roughly speaking, a coalition of $\alpha\%$ of the voters can block a committee if they can find a selection $S$ of $\alpha\% \cdot k$ candidates such that every coalition member has more approved candidates in $S$ than in the committee. 
To date, it is unknown whether the core is always non-empty. The core implies EJR, but all known EJR rules fail the core.
We prove that no welfarist rule can satisfy the core property. Thus, in contrast to extended justified representation, the core is a non-welfarist property. While it remains an open question whether a core committee always exists, this result shows a kind of barrier to an existence proof.

On the other hand, we show that PAV provides a multiplicative approximation to the core: there does not exist a deviation in which all coalition members increase their utility by a factor of more than 2. Further, we show that every rule that provides a better approximation of the core property than 2, must fail the Pigou--Dalton principle of transfers, which formalizes a minimal kind of egalitarianism. Thus, according to our approximation measure, PAV comes closest to satisfy the core among rules that are fair according to the Pigou--Dalton requirement. 

While PAV fails laminar proportionality and priceability, it satisfies EJR. Phragm\'en's rule fails EJR. Might there be an incompatibility between these axioms? The answer is no, and we introduce an attractive committee rule, which we call the \myemph{Method of Equal Shares},\footnote{In earlier versions of this paper, the rule was called Rule X.} which satisfies all three properties. 
Further, it is computable in polynomial time, and is the first non-artificial polynomial-time rule satisfying EJR. In fact, the rule satisfies a strengthened axiom which is intermediate between the core and EJR. We also show that the rule provides a logarithmic approximation of the core property; our analysis is tight. 

The properties of the studied rules are summarized in \Cref{tab:properties_summary}. We complement these results with an additional discussion provided in the appendix. In \Cref{sec:overlapping_parties} we discuss a family of profiles that illustrates a difference between particular priceable rules, in particular between our new method and Phragm\'en's rule. In \Cref{sec:proportionality_and_disagreements} we provide an additional, and much more informal, discussion on the difference between the two types of proportionality offered by PAV and laminar proportional rules.

\section{The Model and Definitions of Rules}

For each $i \in \naturals$, we write $[i] = \{1, \ldots, i\}$. 

\subsection{Elections and preferences}

Given a set of \myemph{candidates} $C = \{c_1, \ldots, c_m\}$ and a set of \myemph{voters} $N = \{1, 2, \ldots, n\}$, an \myemph{approval preference profile} (or \myemph{profile}) is a list $P = (A_1,\dots,A_n)$, where $A_i \subseteq C$ is the approval ballot of the $i$-th voter, i.e., the set of candidates that voter $i$ finds acceptable. We write $N(c)$ for the set of voters who approve candidate $c$. A profile $P = (A_1,\dots,A_n)$ is \myemph{unanimous} if all approval ballots in $P$ are equal. If $P_1$ and $P_2$ are two profiles, then $P_1 + P_2$ is the profile obtained by concatenating the lists $P_1$ and $P_2$. 

An \myemph{election instance} (or \myemph{instance}) is a quadruple $(C, N, P, k)$, where $C$ is a set of candidates, $N$ is a set of voters, $P$ is a preference profile, and $k$ is the desired size of a committee to be elected. We will write $(P, k)$ for an instance when $C$ and $N$ are clear; they can usually be deduced from the profile: $N = \{i : A_i \in P\}$ and $C = \bigcup_{i \in N} A_i$. We will sometimes write $C(P)$ for the set of candidates on which $P$ is defined. We assume that $|C| \geq k$. A \myemph{committee} is a subset $W\subseteq C$ of the candidates with $|W| \le k$. Given a committee $W$, we say that a committee member $c \in W$ is a \myemph{representative} of voter $i$ if $c \in W \cap A_i$.

\subsection{Election rules}

An approval-based \myemph{committee rule} is a function $f$ that, given an instance $(P,k)$, returns a non-empty set of winning committees, each of size at most $k$.\footnote{It is often assumed that a committee rule can only return committees of size exactly $k$. We allow committees of size smaller than $k$ for technical reasons: it will make the definition of our new Method of Equal Shares more natural.} 
Typically, the output of $f$ will be singleton, but several committees may be tied.
Below, we recall definitions of two prominent election rules known in the literature, and introduce one additional rule which is new to this paper. 

\begin{description}
\item[Proportional Approval Voting (PAV).] Given an election instance $(C, N, P, k)$, for each committee $W \subseteq C$ we define its PAV score as:
\begin{align*}
\score_{\mathrm{PAV}}(W) = \sum_{i \in N} \score_{\mathrm{PAV}}(i, W), \quad \text{where} \quad \score_{\mathrm{PAV}}(i, W) = 1 + \frac12 + \frac13 + \cdots + \frac{1}{|W\cap A_i|} \text{.}
\end{align*}
PAV picks committees with the highest score: $\argmax_{W \subseteq C : |W| \le k}\score_{\mathrm{PAV}}(W)$.

\item[Phragm\'en's Sequential Rule.] Assume the voters continuously earn money with the speed of one dollar per time unit (the time is continuous). In the first time moment $t$ when there is a group of voters $S$ who all approve a not-yet-selected candidate $c$, and who have $\nicefrac{n}{k}$ dollars in total, the rule adds $c$ to the committee and asks the voters from $S$ to pay the total amount of $\nicefrac{n}{k}$ for $c$ (that is, the rule resets the balance of each voter from $S$); the other voters keep their so-far earned money. The rule stops when $k$ candidates are selected.\footnote{There exist equivalent definitions of this rule, for example in terms of ``voter load'', which are more in the literature~\citep{aaai/BrillFJL17-phragmen,Janson16arxiv}, but less intuitive and less convenient for our purposes.}

\begin{example}\label{ex:phragmen_definition}
Consider the following approval-based profile and let $k = 4$.
		
		
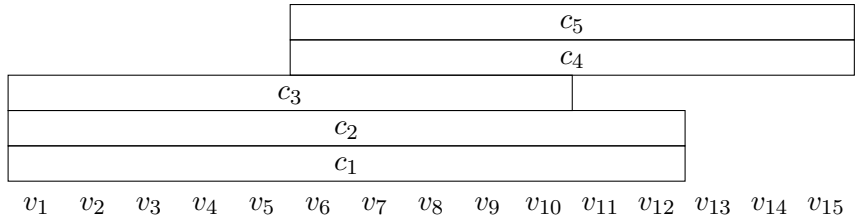
\begin{center}
			\begin{tikzpicture}
			[yscale=0.47,xscale=0.75,voter/.style={anchor=south, yshift=-3pt}, select/.style={fill=green!10}, c/.style={anchor=south, yshift=2.5pt, inner sep=0}]
				\draw (0,0) rectangle (12,1);
				\draw (0,1) rectangle (12,2);
				\draw (0,2) rectangle (10,3);
				\draw (5,3) rectangle (15,4);
				\draw (5,4) rectangle (15,5);
				\node[c] at (6,0) {$c_1$};
				\node[c] at (6,1) {$c_2$};
				\node[c] at (5,2) {$c_3$};
				\node[c] at (10,3) {$c_4$};
				\node[c] at (10,4) {$c_5$};
				\node[voter] at (0.5,-1) {$v_1$};
				\node[voter] at (1.5,-1) {$v_2$};
				\node[voter] at (2.5,-1) {$v_3$};
				\node[voter] at (3.5,-1) {$v_4$};
				\node[voter] at (4.5,-1) {$v_5$};
				\node[voter] at (5.5,-1) {$v_6$};
				\node[voter] at (6.5,-1) {$v_7$};
				\node[voter] at (7.5,-1) {$v_8$};
				\node[voter] at (8.5,-1) {$v_9$};
				\node[voter] at (9.5,-1) {$v_{10}$};
				\node[voter] at (10.5,-1) {$v_{11}$};
				\node[voter] at (11.5,-1) {$v_{12}$};
				\node[voter] at (12.5,-1) {$v_{13}$};
				\node[voter] at (13.5,-1) {$v_{14}$};
				\node[voter] at (14.5,-1) {$v_{15}$};
			\end{tikzpicture}
		 \end{center}
Here, $c_1$ and $c_2$ are approved by the first 12 voters, $c_3$ is approved by the first 10 voters, and $c_4$ and $c_5$ are approved by the last 10 voters. The way Phragm\'en's Sequential Rule operates in this profile is depicted in the following figure (in the figure, the area of each polygon corresponding to a selected candidate equals $\nicefrac{15}{4}$).
		
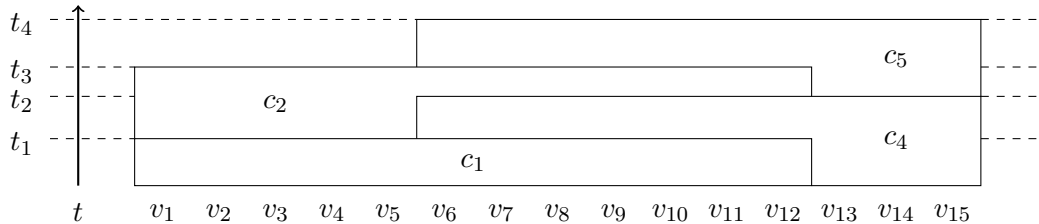
\begin{center}
			\begin{tikzpicture}
			[yscale=2.0,xscale=0.75,voter/.style={anchor=south, yshift=-3pt}, select/.style={fill=green!10}, c/.style={anchor=south, yshift=2.5pt, inner sep=0}]
				\draw (0,0) rectangle (12,0.3125);
				\draw (12,0) -- (15,0) -- (15,0.59375) -- (5,0.59375) -- (5,0.3125); 
                                 \draw (0,0.3125) -- (0,0.7890625) -- (12,0.7890625) -- (12,0.59375); 
                                 \draw (15,0.59375) -- (15,1.10546875) -- (5,1.10546875) -- (5,0.7890625); 
                                 
                                 \node[c] at (6,0.05) {$c_1$};
                                 \node[c] at (13.5,0.2) {$c_4$};
                                 \node[c] at (2.5,0.44) {$c_2$};
                                  \node[c] at (13.5,0.74) {$c_5$};
                                 
                                 \draw[thick,->] (-1, 0.0) -- (-1, 1.2);
                                 
                                 \node[voter] at (-1,  -0.25) {$t$};
                                	\node[voter] at (0.5,-0.25) {$v_1$};
				\node[voter] at (1.5,-0.25) {$v_2$};
				\node[voter] at (2.5,-0.25) {$v_3$};
				\node[voter] at (3.5,-0.25) {$v_4$};
				\node[voter] at (4.5,-0.25) {$v_5$};
				\node[voter] at (5.5,-0.25) {$v_6$};
				\node[voter] at (6.5,-0.25) {$v_7$};
				\node[voter] at (7.5,-0.25) {$v_8$};
				\node[voter] at (8.5,-0.25) {$v_9$};
				\node[voter] at (9.5,-0.25) {$v_{10}$};
				\node[voter] at (10.5,-0.25) {$v_{11}$};
				\node[voter] at (11.5,-0.25) {$v_{12}$};
				\node[voter] at (12.5,-0.25) {$v_{13}$};
				\node[voter] at (13.5,-0.25) {$v_{14}$};
				\node[voter] at (14.5,-0.25) {$v_{15}$};
				
				\draw[dashed] (-1.5, 0.3125) -- (0.0, 0.3125);
				\draw[dashed] (15, 0.3125) -- (16, 0.3125);
				\draw[dashed] (-1.5, 0.59375) -- (0.0, 0.59375);
				\draw[dashed] (15, 0.59375) -- (16, 0.59375);
				\draw[dashed] (-1.5, 0.7890625) -- (0.0, 0.7890625);
				\draw[dashed] (15, 0.7890625) -- (16, 0.7890625);
                                 \draw[dashed] (-1.5, 1.10546875) -- (5.0, 1.10546875);
                                 \draw[dashed] (15, 1.10546875) -- (16, 1.10546875);
                                 
                                 \node[voter] at (-2,0.3125 - 0.12) {$t_1$};
                                 \node[voter] at (-2,0.59375 - 0.12) {$t_2$};
                                 \node[voter] at (-2,0.7890625 - 0.12) {$t_3$};
                                 \node[voter] at (-2,1.10546875 - 0.12) {$t_4$};
			\end{tikzpicture}
		 \end{center} 
		 Candidate $c_1$ will be selected first in time $t_1 = \nicefrac{15}{4 \cdot 12}$, and all 12 voters who approve $c_1$ will be charged $\nicefrac{15}{48}$ dollars for that. The remaining 3 voters are left with the total amount of $3 \cdot  \nicefrac{15}{48}$ dollars. At time $t_2 = t_1 + \Delta$ such that:
		 $3 \cdot  \nicefrac{15}{48} + 10\Delta = \nicefrac{15}{48},$
		 the last 10 voters will be able to afford to buy candidate $c_4$. The first 12 voters would be able to afford $c_2$ only at time $t_1 + \nicefrac{15}{48}$. Since $\Delta = \nicefrac{9}{32} < \nicefrac{15}{48}$, the next candidate selected will be $c_4$. Voters $v_6 \ldots v_{12}$ will pay $\nicefrac{9}{32}$ for $c_2$. Each of the last 3 voters will pay $\nicefrac{15}{48} + \nicefrac{9}{32}$. After selecting $c_4$ each of the last 10 voters is left with no money, and each of the first 5 voters has $\nicefrac{9}{32}$ dollars left. Next, $c_2$ will be selected in time $t_3 = t_2 + \nicefrac{25}{128}$, and $c_5$ at time $t_4 = t_3 + \nicefrac{81}{256}$. The selected committee is $\{c_1, c_2, c_4, c_5\}$.  \qed
\end{example}

\item[The Method of Equal Shares.]\!\!\!\footnote{In earlier versions of this paper, the rule was called Rule X. We believe that the new name is more informative. The part ``equal shares'' corresponds to two elements in the definition of the rule. First, the budget is divided equally among the voters. Second, once a candidate is selected, its cost is shared as equally as possible among the voters who approve the candidate.} (For short, we often call this rule \emph{Equal Shares}.) This rule is new to the paper. Each voter $i \in N$ has an initial budget of one dollar, $b_i(1) = 1$; the voters spend their money during the run of the rule, buying candidates they approve of. Buying a candidate costs $p = \nicefrac{n}{k}$ dollars in total (we call $p$ the \myemph{price}). The rule starts with an empty committee $W = \emptyset$ and adds candidates to~$W$ sequentially; let $b_i(t)$ be the amount of money that voter $i$ is left with just before the start of the $t$-th iteration. In the $t$-th step, the rule selects the candidate that should be added to~$W$ as follows. For a value $q \ge 0$, we say that a candidate $c \not \in W$ is \emph{$q$-affordable} at round $t$ if
\begin{align}
\label{eq:afford}
\sum_{i \in N(c)} \min(q, b_i(t)) \ge p \text{.}
\end{align}
Thus, the voters approving $c$ can raise the amount $p$ required to elect $c$ while each paying at most $q$. If no candidate $c \not\in W$ is $q$-affordable for any $q$, the rule stops and returns $W$. Otherwise, the rule selects a candidate $c\not\in W$ which is $q$-affordable for a minimum $q$, and adds $c$ to the committee $W$. Note that by minimality of $q$, inequality \eqref{eq:afford} holds with equality. For each voter $i \in N(c)$, we set their budget to $b_i(t+1) = b_i(t) - \min(q, b_i(t))$. (So each of these voters pays either $q$ or their entire remaining budget for $c$.) For each $i \notin N(c)$ we set $b_i(t+1) = b_i(t)$.\footnote{Since for each selected committee member the voters pay $\nicefrac{n}{k}$ dollars in total, the rule picks at most $k$ candidates. The rule may select strictly fewer than $k$ candidates. If desired, the committee can then be completed using different strategies.  Most of the crucial properties of Equal Shares (e.g., EJR and laminar proportionality) do not depend on the particular strategy. For priceability, the choice of strategy matters.}

This rule is similar to Phragm\'en's Sequential Rule, except in Equal Shares we give voters money up front, and in Phragm\'en's rule we give it to them continuously.

\begin{example}
Consider the profile from \Cref{ex:phragmen_definition}. Recall that $k = 4$. The way the Method of Equal Shares proceeds given this profile is depicted below. The diagram also illustrates the minimal values~$q$ for which the respective elected candidates are $q$-affordable.  

		
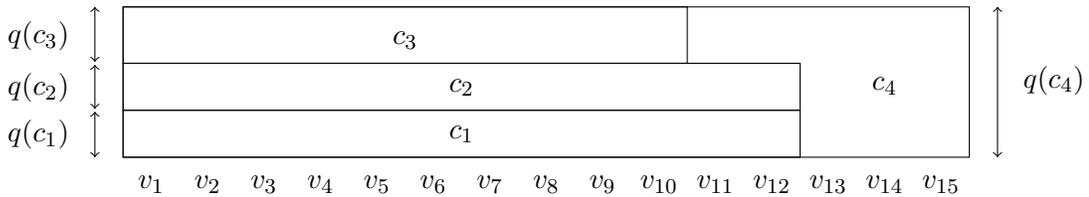
\begin{center}
			\begin{tikzpicture}
			[yscale=2.0,xscale=0.75,voter/.style={anchor=south, yshift=-3pt}, select/.style={fill=green!10}, c/.style={anchor=south, yshift=2.5pt, inner sep=0}]
				\draw (0,0) rectangle (12,0.3125);
				\draw (0,0.3125) rectangle (12,0.625);
				\draw (0,0.625) rectangle (10,1.0);
				\draw (0,0) rectangle (15,1.0);
			                           
                                 \node[c] at (6,0.05) {$c_1$};
                                 \node[c] at (6,0.3625) {$c_2$};
                                 \node[c] at (5,0.675) {$c_3$};
                                 \node[c] at (13.5,0.37) {$c_4$};
                                
                                	\node[voter] at (0.5,-0.25) {$v_1$};
				\node[voter] at (1.5,-0.25) {$v_2$};
				\node[voter] at (2.5,-0.25) {$v_3$};
				\node[voter] at (3.5,-0.25) {$v_4$};
				\node[voter] at (4.5,-0.25) {$v_5$};
				\node[voter] at (5.5,-0.25) {$v_6$};
				\node[voter] at (6.5,-0.25) {$v_7$};
				\node[voter] at (7.5,-0.25) {$v_8$};
				\node[voter] at (8.5,-0.25) {$v_9$};
				\node[voter] at (9.5,-0.25) {$v_{10}$};
				\node[voter] at (10.5,-0.25) {$v_{11}$};
				\node[voter] at (11.5,-0.25) {$v_{12}$};
				\node[voter] at (12.5,-0.25) {$v_{13}$};
				\node[voter] at (13.5,-0.25) {$v_{14}$};
				\node[voter] at (14.5,-0.25) {$v_{15}$};
				
				\draw[<->] (-0.5, 0.0) -- (-0.5, 0.3025);
				\draw[<->] (-0.5, 0.3225) -- (-0.5, 0.615);
				\draw[<->] (-0.5, 0.635) -- (-0.5, 1.0);
				\draw[<->] (15.5, 0.0) -- (15.5, 1.0);
				
				\node[voter] at (-1.5,0.04) {$q(c_1)$};
				\node[voter] at (-1.5,0.3525) {$q(c_2)$};
				\node[voter] at (-1.5,0.695) {$q(c_3)$};
				\node[voter] at (16.5,0.4) {$q(c_4)$};
			\end{tikzpicture}
		 \end{center}
	In the first step, $c_1$ is $\nicefrac{15}{48}$-affordable and each of the 12 voters who approve $c_1$ pays $\nicefrac{15}{48}$. In the second step, $c_2$ is $\nicefrac{15}{48}$-affordable and again each of the first 12 voters pays $\nicefrac{15}{48}$. Next, $c_3$, $c_4$, and $c_5$ are all $\nicefrac{15}{40}$-affordable. If $c_4$ is selected in the third step, then no other candidates will be affordable, and so Equal Shares selects the committee $\{c_1, c_2, c_4\}$. For the sake of this example, assume that $c_3$ is selected in the third step. Then, in the fourth step each of the first 10 voters is left with no money, voters $v_{11}$ and $v_{12}$ have $\nicefrac{18}{48}$ dollars left, and the last 3 voters have 1 dollar each. Thus, both $c_4$ and $c_5$ are 1-affordable, and one of them (say $c_4$) is selected. Then no candidate is affordable. The selected committee is $\{c_1, c_2, c_3, c_4\}$.  \qed
\end{example}

\end{description}

It is clear that Phragm\'en's rule and the Method of Equal Shares can be calculated in polynomial time.\footnote{For Equal Shares, to calculate the minimum $q$ that makes a candidate $q$-affordable, sort the agents $N(c)$ approving $c$ in increasing order of budget $b_i(t)$ remaining. Then for each $j = 0, \dots, |N(c)|$, check whether we can raise the required amount $p$ by having the poorest $j$ agents contribute their entire remaining budget, and having the richest $|N(c)|-j$ agents split the remaining amount equally.}
On the other hand, PAV is NP-hard to evaluate \citep{azi-gas-gud-mac-mat-wal:c:multiwinner-approval}.

\subsection{Welfarist rules}

We are particularly interested in a class of rules that maximize an objective function based on voters' utilities. We call such rules \emph{welfarist}, a concept from welfare economics \citep{sen1979utilitarianism}.

\begin{definition}[Welfarism]\label{def:welfarist}
Given a profile $P$, we define the welfare vector of a committee $W$ as $w_P(W) = (|A_1 \cap W|, \dots, |A_n \cap W|)$.
A committee rule $f$ is \myemph{welfarist} if for each $k$ there exists a function $g_k$ mapping welfare vectors to real values, such that for each instance $(P, k)$ we have:
\begin{align*}
f(P, k) = \argmax_{W \subseteq C : |W| = k} g_k(w_P(W)) \text{.}
\end{align*}  
\end{definition}

For example, $g_k$ could be summation, and then $f$ is the utilitarian rule. If we encode the leximin principle in $g_k$, we obtain an approval-variant of the rule of \citet{ccElection}. Clearly, PAV is a welfarist rule, and so are all variants of PAV that use scoring vectors other than harmonic numbers (see, e.g., \citealp{lac-sko:t:approval-thiele}). Serial dictatorships can also be captured.

We say that a welfarist rule is \myemph{Pareto-optimal} if the functions $g_k$ are monotonic, and it is easy to see that the corresponding committee rule is then Pareto-optimal in the usual sense.

\section{Laminar Proportionality}
	
	In this section, we define a new axiom called \myemph{laminar proportionality}. This axiom identifies a class of well-behaved profiles (\myemph{laminar profiles}) and specifies which committees are acceptable outputs on these profiles. The definition of the axiom is a bit tedious, so we start by giving examples which will illustrate the ideas.
	The first example also recalls the graphical representation of profiles that we will use throughout the paper.
	
	\subsection{Examples of laminar profiles}
	
	\begin{example}
		[Integral party-list profiles]
		\label{ex:laminar-on-party-lists}
		A profile $P=(A_1, \dots, A_n)$ is a \myemph{party-list profile} if for all $i,j \in N$, either $A_i = A_j$ or $A_i \cap A_j = \emptyset$. Thus, there are $r$ parties $C_1, \dots, C_r \subseteq C$ of pairwise disjoint sets of candidates, and each voter approves exactly the members of one party. Write $n_1, \dots, n_r$ for the number of voters supporting each party. We call an election instance $(P,k)$ an \myemph{integral party-list instance} if $P$ is a party-list profile and the values $k \cdot n_1/n,\dots,k\cdot n_r/n$ are all integral. We will also require that $|C_z| \ge k \cdot n_z/n$ for each party $z = 1, \dots, r$.
		\begin{figure}[h]
			\centering
			\begin{tikzpicture}
			[yscale=0.4,xscale=0.75,voter/.style={anchor=south, yshift=-4pt}]
				\foreach \y in {0,...,4} {
					\draw (0,\y) rectangle (3,\y+1);
					\draw (3,\y) rectangle (6,\y+1);
					\draw (6,\y) rectangle (8,\y+1);
				}
				\foreach \y in {0,...,2} {
					\draw[fill=green!10] (0,\y) rectangle (3,\y+1);
					\draw[fill=green!10] (3,\y) rectangle (6,\y+1);
				}
				\draw[fill=green!10] (6,0) rectangle (8,1);
				\draw[fill=green!10] (6,1) rectangle (8,2);
				\node[voter] at (0.5,-1) {$v_1$};
				\node[voter] at (1.5,-1) {$v_2$};
				\node[voter] at (2.5,-1) {$v_3$};
				\node[voter] at (3.5,-1) {$v_4$};
				\node[voter] at (4.5,-1) {$v_5$};
				\node[voter] at (5.5,-1) {$v_6$};
				\node[voter] at (6.5,-1) {$v_7$};
				\node[voter] at (7.5,-1) {$v_8$};
				
				\node[voter] at (1.5,0.25) {$c_1$};
			\end{tikzpicture}
			\caption{An integral party-list instance for $k = 8$. The boxes correspond to the candidates, and each candidate is approved by the voters below the corresponding box. E.g., candidate $c_1$ who is represented by the left bottom box is approved by voters $v_1$, $v_2$, and $v_3$.}
			\label{fig:party-list}
		\end{figure}
		\Cref{fig:party-list} shows an example of an integral party-list instance. Each box denotes a candidate, and each voter $v_i$ approves all candidates appearing in the column for $v_i$. In the example, there are three parties, where $v_1,v_2, v_3$ support party $C_1$, while $v_4,v_5,v_6$ support $C_2$, and $v_7, v_8$ support $C_3$. Every such instance is a laminar profile. Our axiom of laminar proportionality will require that a committee of size $k=8$ wins only if it contains three candidates from $C_1$, three candidates from $C_2$, and two candidates from $C_3$. An example of such a committee is indicated in green. 
		\qed
	\end{example}
	
	When introducing PAV, \citet{Thie95a} discussed exactly the property from \Cref{ex:laminar-on-party-lists}, that is, proportionality on integral party-list profiles. He showed that PAV satisfies it. Our axiom of laminar proportionality will strengthen this axiom by enlarging the class of profiles on which we specify the outcome, and PAV will fail this strengthening. Our larger class of laminar profiles keeps the spirit of `integral' profiles (and sidestep rounding issues), but goes beyond party-lists. \citet{bri-las-sko:c:apportionment} and \citet{lac-sko:t:approval-thiele} study axioms that apply to party-list profiles without the integrality requirement (see also \Cref{sec:price}), but as we discuss in \Cref{sec:additional_discussion} (\Cref{rem:integrality_in_laminarity_def}) the integrality must be required for our larger class of laminar profiles.
	
	\begin{example}
		[Two parties with a common leader]
		\label{ex:laminar-two-parties}
		Consider the following instance with $k = 4$.
		
		\begin{center}
			\begin{tikzpicture}
			[yscale=0.47,xscale=0.75,voter/.style={anchor=south, yshift=-3pt}, select/.style={fill=green!10}, c/.style={anchor=south, yshift=2.5pt, inner sep=0}]
				\draw[select] (0,0) rectangle (6,1);
				\draw[select] (0,1) rectangle (4,2);
				\draw[select] (0,2) rectangle (4,3);
				\draw (0,3) rectangle (4,4);
				\draw[select] (4,1) rectangle (6,2);
				\draw (4,2) rectangle (6,3);
				\draw (4,3) rectangle (6,4);
				\draw (4,4) rectangle (6,5);
				\node[c] at (3,0) {$c_1$};
				\node[c] at (2,1) {$c_2$};
				\node[c] at (2,2) {$c_3$};
				\node[c] at (2,3) {$c_4$};
				\node[c] at (5,1) {$c_5$};
				\node[c] at (5,2) {$c_6$};
				\node[c] at (5,3) {$c_7$};
				\node[c] at (5,4) {$c_8$};
				\node[voter] at (0.5,-1) {$v_1$};
				\node[voter] at (1.5,-1) {$v_2$};
				\node[voter] at (2.5,-1) {$v_3$};
				\node[voter] at (3.5,-1) {$v_4$};
				\node[voter] at (4.5,-1) {$v_5$};
				\node[voter] at (5.5,-1) {$v_6$};
			\end{tikzpicture}
		\end{center}
		In this instance, there is one candidate $c_1$ who is approved by all voters. Otherwise, the voters are partitioned into supporters of two disjoint ``parties'', with four voters supporting $\{c_2, c_3, c_4\}$ and two voters supporting $\{c_5, c_6, c_7, c_8\}$. This example is essentially the example discussed in the introduction. Laminar proportionality requires that a winning committee of size $k=4$ includes $c_1$, and that the remaining three seats are filled in proportion to the party supports: two seats for $\{c_2, c_3, c_4\}$ and one seat for $\{c_5, c_6, c_7, c_8\}$. The figure indicates in green an example of a committee that is laminar proportional.
		\qed
	\end{example}
	
	\begin{example}
		[Subdivided parties]
		\label{ex:laminar-subdivided-parties}
		Consider the following instance with $k = 12$.
		
		\begin{center}
			\begin{tikzpicture}
			[yscale=0.43,xscale=0.75,voter/.style={anchor=south, yshift=-4pt}, select/.style={fill=green!10}, c/.style={anchor=south, yshift=1.5pt, inner sep=0}]
				\draw[select] (0,0) rectangle (6,1);
				\draw[select] (0,1) rectangle (6,2);
				\draw[select] (0,2) rectangle (6,3);
				\draw[select] (0,3) rectangle (6,4);
				\draw[select] (0,4) rectangle (3,5);
				\draw[select] (0,5) rectangle (3,6);
				\draw[select] (3,4) rectangle (6,5);
				\draw[select] (3,5) rectangle (6,6);
				\draw (3,6) rectangle (6,7);
				\draw (3,7) rectangle (6,8);
				
				\draw[select] (6,0) rectangle (9,1);
				\draw[select] (6,1) rectangle (8,2);
				\draw[select] (6,2) rectangle (8,3);
				\draw (6,3) rectangle (8,4);
				\draw (6,4) rectangle (8,5);
				\draw (6,5) rectangle (8,6);
				\draw (6,6) rectangle (8,7);
				\draw[select] (8,1) rectangle (9,2);
				\draw (8,2) rectangle (9,3);
				\draw (8,3) rectangle (9,4);
				\foreach \i in {1,...,9}
					\node[voter] at (\i-0.5,-1) {$v_\i$};
			\end{tikzpicture}
		\end{center}
		In this instance, there are two major parties. The left-hand party is supported by two thirds of the voters ($v_1,\dots,v_6$) and thus deserves $\frac23 \cdot k = 8$ seats. The left-hand party has four candidates approved by all its supporters, but is otherwise evenly divided into a left wing and a right wing. Laminar proportionality will require that the four consensus candidates are elected, and that the remaining 4 seats for the left-hand party be evenly divided among the two wings. There is also a right-hand party, deserving $\frac13 \cdot k = 4$ seats, of which one will go to the consensus candidate of that party, with the remaining 3 seats divided proportionally between that party's wings. The figure indicates in green an example of a laminar proportional committee.
		\qed
	\end{example}
	
	We also give some intuitive explanation of the difference between laminar proportionality and ``welfarist'' proportionality in \Cref{sec:proportionality_and_disagreements}.
	
	\subsection{Definition and properties}
	
	\begin{definition}
		An election instance $(P,k)$ is \myemph{laminar}\footnote{The term ``laminar'' refers to a property of set systems: a family of sets is \myemph{laminar} if any two sets in the family are either disjoint or one is a subset of the other. For a laminar election instance, the set system $\{N_c\}_{c \in C}$, where $N_c$ is the set of voters approving $c$, is laminar (see \Cref{prop:laminarity_of_set_system} in the appendix).}
		if either:
		\begin{itemize}
			\item[(i)] $P$ is unanimous and $|C(P)| \ge k$.
			\item[(ii)] There is a candidate $c \in C(P)$ such that $c \in A_i$ for all $A_i \in P$, the profile $P_{-c}$ is not unanimous, and the instance $(P_{-c}, k-1)$ with $P_{-c} = (A_1\setminus\{c\}, \dots, A_n\setminus\{c\})$ is laminar.
			\item[(iii)] There are two laminar instances $(P_1, k_1)$ and $(P_2, k_2)$ with $C(P_1) \cap C(P_2) = \emptyset$ and $|P_1|/k_1 = |P_2|/k_2$ such that $P = P_1 + P_2$ and $k = k_1 + k_2$.\footnote{In follow-up work, \citet{los2022systematization} implicitly allow $k = 0$ throughout, and avoid division by zero in part (iii) by writing the condition as $|P_1|k_2 = |P_2|k_1$. This slightly enlarges the space of laminar instances, and can be useful.}
		\end{itemize}
	\end{definition}
	
	Integral party-list instances (discussed in \Cref{ex:laminar-on-party-lists}) are laminar, since they can be obtained by forming a disjoint sum (iii) of unanimous profiles (i). Party-list instances that are not integral are not laminar, since we cannot use operation (iii). \Cref{ex:laminar-two-parties} is laminar, since it can be obtained by building a party-list instance as before, and then adding the unanimously approved candidate~$c_1$ using operation (ii). \Cref{ex:laminar-subdivided-parties} is laminar by applying (iii) to two instances obtained similarly to \Cref{ex:laminar-two-parties}.
	
	\begin{figure}[t!]
		\begin{center}
			\begin{tikzpicture}
				[yscale=0.387,xscale=0.702,every node/.style={scale=0.9},voter/.style={anchor=south, yshift=-7pt}, select/.style={fill=blue!10}, c/.style={anchor=south, yshift=1.5pt, inner sep=0}]
				
				\node at (-3.0, 31) {\Cref{def:laminar_proportionality} (1)};
				\draw[select] (-3.5,29) rectangle (-2.5,30);
				\node at (-3.0,29.42) {$c_5$};
				\node[voter] at (-3.0,28) {$v_2$};
				\node[voter] at (-3.0,27) {$k=1$};
				
				\node at (5.0, 31) {\Cref{def:laminar_proportionality} (1)};
				\draw[select] (4.5,29) rectangle (5.5,30);
				\node at (5.0,29.42) {$c_6$};
				\node[voter] at (5.0,28) {$v_3$};
				\node[voter] at (5.0,27) {$k=1$};
				
				\draw[thick,->] (-2.0, 29) -- (1.0, 25);
				\draw[thick,->] (4.0, 29) -- (2.0, 25);
				\node at (1.2, 28.2) {\Cref{def:laminar_proportionality} (3)};

				\node at (-6.5, 25) {\Cref{def:laminar_proportionality} (1)};
				\draw[select] (-7,23) rectangle (-6,24);
				\node at (-6.5,23.42) {$c_4$};
				\node[voter] at (-6.5,22) {$v_1$};
				\node at (-6.5, 21) {$k = 1$};
				
				\draw[select] (0.5,23) rectangle (1.5,24);
				\draw[select] (1.5,23) rectangle (2.5,24);
				\node at (1.0,23.42) {$c_5$};
				\node at (2.0,23.42) {$c_6$};
				\node[voter] at (1.0,22) {$v_2$};
				\node[voter] at (2.0,22) {$v_3$};
				\node at (1.5, 21) {$k = 2$};
				
				\draw[thick,->] (-5.5, 23) -- (-3, 20);
				\draw[thick,->] (0.0, 23) -- (-2, 20);
				\node at (-2.5, 22.8) {\Cref{def:laminar_proportionality} (3)};

				\draw[select] (-4,18) rectangle (-3,19);
				\draw[select] (-3,18) rectangle (-2,19);
				\draw[select] (-2,18) rectangle (-1,19);
				\node at (-3.5,18.42) {$c_4$};
				\node at (-2.5,18.42) {$c_5$};
				\node at (-1.5,18.42) {$c_6$};
				\foreach \i in {1,...,3}
				\node[voter] at (\i-4.5,17) {$v_{\i}$};
				\node at (-2.5, 16) {$k = 3$};
				
				\draw[thick] (-2.5, 15.0) -- (-2.5, 14.2);
				\node at (-2.5, 13.5) {\Cref{def:laminar_proportionality} (2)};
				\draw[thick,->] (-2.5, 12.8) -- (-2.5, 11.5);

				\node at (2.5, 18) {\Cref{def:laminar_proportionality} (1)};
				\draw[select] (2.0,14) rectangle (3.0,15);
				\draw[select] (2.0,15) rectangle (3.0,16);
				\draw (2.0,16) rectangle (3.0,17);
				\node at (2.5,14.42) {$c_{7}$};
				\node at (2.5,15.42) {$c_{8}$};
				\node at (2.5,16.42) {$c_{9}$};
				\node[voter] at (2.5,13) {$v_4$};
				\node[voter] at (2.5,12) {$k=2$};

				\node at (6.5, 24) {\Cref{def:laminar_proportionality} (1)};
				\draw[select] (6.0,20) rectangle (7.0,21);
				\draw[select] (6.0,21) rectangle (7.0,22);
				\draw (6.0,22) rectangle (7.0,23);
				\node at (6.5,20.42) {$c_{10}$};
				\node at (6.5,21.42) {$c_{11}$};
				\node at (6.5,22.42) {$c_{12}$};
				\node[voter] at (6.5,19) {$v_5$};
				\node[voter] at (6.5,18) {$k=2$};
				
				\node at (12.5, 27) {\Cref{def:laminar_proportionality} (1)};
				\draw[select] (12.0,23) rectangle (13.0,24);
				\draw[select] (12.0,24) rectangle (13.0,25);
				\draw (12.0,25) rectangle (13.0,26);
				\node at (12.5,23.42) {$c_{13}$};
				\node at (12.5,24.42) {$c_{14}$};
				\node at (12.5,25.42) {$c_{15}$};
				\node[voter] at (12.5,22) {$v_6$};
				\node[voter] at (12.5,21) {$k=2$};
				
				\draw[thick,->] (12.5, 20) -- (12.5, 18);
				\draw[thick,->] (7.5, 20) -- (11.5, 18);

				\draw[select] (11.0,14) rectangle (12.0,15);
				\draw[select] (12.0,14) rectangle (13.0,15);
				\draw[select] (11.0,15) rectangle (12.0,16);
				\draw[select] (12.0,15) rectangle (13.0,16);
				\draw (11.0,16) rectangle (12.0,17);
				\draw (12.0,16) rectangle (13.0,17);
				\node at (11.5,14.42) {$c_{10}$};
				\node at (11.5,15.42) {$c_{11}$};
				\node at (11.5,16.42) {$c_{12}$};
				\node at (12.5,14.42) {$c_{13}$};
				\node at (12.5,15.42) {$c_{14}$};
				\node at (12.5,16.42) {$c_{15}$};
				\node[voter] at (11.5,13) {$v_5$};
				\node[voter] at (12.5,13) {$v_6$};
				\node[voter] at (12.0,12) {$k=4$};

				\draw[thick,->] (3.5, 14) -- (8.0, 11);
				\draw[thick,->] (10.5, 14) -- (9.0, 11);
				\node at (7.8, 13.2) {\Cref{def:laminar_proportionality} (3)};

				\draw[select] (-4,7) rectangle (-1,8);
				\draw[select] (-4,8) rectangle (-1,9);
				\draw[select] (-4,9) rectangle (-1,10);
				\draw[select] (-4,10) rectangle (-3,11);
				\draw[select] (-3,10) rectangle (-2,11);
				\draw[select] (-2,10) rectangle (-1,11);
				\node at (-2.5,7.42) {$c_1$};
				\node at (-2.5,8.42) {$c_2$};
				\node at (-2.5,9.42) {$c_3$};
				\node at (-3.5,10.42) {$c_4$};
				\node at (-2.5,10.42) {$c_5$};
				\node at (-1.5,10.42) {$c_6$};
				\foreach \i in {1,...,3}
				\node[voter] at (\i-4.5,6) {$v_{\i}$};
				\node at (-2.5, 5) {$k = 6$};
				
				\foreach \x in {3,4,5}
				{
					\foreach \y in {0,1}
					{
						\draw[select] (\x+4,\y+7) rectangle (\x+5,\y+8);
					}
					\foreach \y in {2}
					{
						\draw (\x+4,\y+7) rectangle (\x+5,\y+8);
					}
				}
				\node at (7.5,7.42) {$c_{7}$};
				\node at (7.5,8.42) {$c_{8}$};
				\node at (7.5,9.42) {$c_{9}$};
				\node at (8.5,7.42) {$c_{10}$};
				\node at (8.5,8.42) {$c_{11}$};
				\node at (8.5,9.42) {$c_{12}$};
				\node at (9.5,7.42) {$c_{13}$};
				\node at (9.5,8.42) {$c_{14}$};
				\node at (9.5,9.42) {$c_{15}$};
				\foreach \i in {4,...,6}
				\node[voter] at (\i+3.5,6) {$v_{\i}$};
				\node at (8.5, 5) {$k = 6$};
				
				\draw[thick,->] (-0.5, 7) -- (2, 5);
				\draw[thick,->] (6.5, 7) -- (4, 5);  
				\node at (3, 7.5) {\Cref{def:laminar_proportionality} (3)};

				\draw[select] (0,0) rectangle (3,1);
				\draw[select] (0,1) rectangle (3,2);
				\draw[select] (0,2) rectangle (3,3);
				\draw[select] (0,3) rectangle (1,4);
				\draw[select] (1,3) rectangle (2,4);
				\draw[select] (2,3) rectangle (3,4);
				\node at (1.5,0.42) {$c_1$};
				\node at (1.5,1.42) {$c_2$};
				\node at (1.5,2.42) {$c_3$};
				\node at (0.5,3.42) {$c_4$};
				\node at (1.5,3.42) {$c_5$};
				\node at (2.5,3.42) {$c_6$};
				\foreach \x in {3,4,5}
				{
					\foreach \y in {0,1}
					{
						\draw[select] (\x,\y) rectangle (\x+1,\y+1);
					}
					\foreach \y in {2}
					{
						\draw (\x,\y) rectangle (\x+1,\y+1);
					}
				}
				\node at (3.5,0.42) {$c_{7}$};
				\node at (3.5,1.42) {$c_{8}$};
				\node at (3.5,2.42) {$c_{9}$};
				\node at (4.5,0.42) {$c_{10}$};
				\node at (4.5,1.42) {$c_{11}$};
				\node at (4.5,2.42) {$c_{12}$};
				\node at (5.5,0.42) {$c_{13}$};
				\node at (5.5,1.42) {$c_{14}$};
				\node at (5.5,2.42) {$c_{15}$};
				\foreach \i in {1,...,6}
				\node[voter] at (\i-0.5,-1) {$v_{\i}$};
				
				\node at (3,-2.0) {$k = 12$};
			\end{tikzpicture}
		\end{center}
		\caption{Graphical representation of the recursive reasoning that the election instance from the introduction is laminar, and that the committees returned by Phragm\'en's sequential rule and by Equal Shares for this instance are laminar proportional.}\label{fig:laminar-proportionality-reasoning}
	\end{figure}
	
	Our axiom of laminar proportionality specifies which committees are allowed outcomes on a laminar instance.
	
	\begin{definition}\label{def:laminar_proportionality}
		A committee $W \subseteq C(P)$ is \myemph{laminar proportional} for a laminar instance $(P,k)$ if $|W| = k$ and
		\begin{enumerate}
			\item If $P$ is unanimous, then $W \subseteq C(P)$.
			\item If there is a unanimously approved candidate $c$ such that $(P_{-c}, k-1)$ is laminar, then $W = W' \cup \{c\}$ where $W'$ is a committee which is laminar proportional for $(P_{-c}, k-1)$.
			\item If $P$ is the sum of laminar instances $(P_1, k_1)$ and $(P_2, k_2)$, then $W = W_1 \cup W_2$ where $W_1$ is laminar proportional for $(P_1, k_1)$ and $W_2$ is laminar proportional for $(P_2, k_2)$.
		\end{enumerate}
		A committee rule $f$ is \myemph{laminar proportional} if for every laminar instance it returns only laminar proportional committees.
	\end{definition}

One can check using a straightforward induction that Phragm\'en's rule satisfies laminar proportionality. The same is true for the Method of Equal Shares. For example, the recursive reasoning showing that the two rules return laminar proportional committees for the election instance from the introduction is depicted in \Cref{fig:laminar-proportionality-reasoning}. 
On the other hand, as we have seen in the introduction, PAV fails laminar proportionality. A detailed proof is given in the appendix.
\newcommand{\phragmensatisfieslaminarproportionality}{\mbox{}
		\begin{enumerate}
			\item[(a)] Phragm\'en's rule is laminar proportional.
			\item[(b)] The Method of Equal Shares is laminar proportional.
			\item[(c)] PAV fails laminar proportionality.
		\end{enumerate}}

\begin{samepage}
	\begin{theorem}\label{thm:phragmen_satisfies_laminar_proportionality} \mbox{}
		\phragmensatisfieslaminarproportionality
	\end{theorem}
\end{samepage}

\subsection{No welfarist rule satisfies laminar proportionality}
We have seen that PAV fails laminar proportionality. The rules that do satisfy it are not welfarist rules. We now establish that no welfarist rule can satisfy the axiom. In our proof, we construct two laminar instances for which the laminar proportional outcomes induce two different welfare vectors. However, both welfare vectors can be induced by some committee in both instances. Hence, a welfarist rule cannot choose the laminar proportional outcome in both cases.
	
	\begin{theorem}\label{thm:no_welfarist_lam_prop}
There exists no welfarist committee rule that satisfies laminar proportionality.
\end{theorem}
\begin{proof}
Assume for a contradiction that there exists a welfarist rule $f$ which satisfies laminar proportionality. Let $g_k$ be as in \Cref{def:welfarist}.

Consider the following laminar instance with $k = 20$. The two figures show two committees; the blue one on the left is laminar proportional, while the green one on the right is not.

\begin{center}
	\begin{minipage}{0.4\linewidth}
		\begin{tikzpicture}
	[yscale=0.47,xscale=0.75,voter/.style={anchor=south, yshift=-3pt}, s/.style={fill=blue!10}, n/.style={}, c/.style={anchor=south, yshift=2.5pt, inner sep=0}]
		\foreach \i / \x / \y / \wid / \pick in 
			{1/0/0/4/s, 2/0/1/4/s, 
			3/4/0/4/s, 4/4/1/4/s,
			5/0/2/2/s, 6/0/3/2/s, 7/0/4/2/s, 8/0/5/2/s, 9/0/6/2/n,
			10/2/2/2/s, 11/2/3/2/s, 12/2/4/2/s, 13/2/5/2/s, 14/2/6/2/n,
			15/4/2/2/s, 16/4/3/2/s, 17/4/4/2/s, 18/4/5/2/s,
			19/6/2/2/s, 20/6/3/2/s, 21/6/4/2/s, 22/6/5/2/s} {
			\draw[\pick] (\x,\y) rectangle (\x + \wid,\y+1);
			\node[c] at (\x + 0.5 * \wid, \y) {$c_{\i}$}; 	
		}
		\foreach \i/\u in {1/6,2/6,3/6,4/6,5/6,6/6,7/6,8/6} {
			\node[voter] at (\i-0.5,-1) {$v_\i$};
			\node[font=\footnotesize, inner sep=2pt, fill=black!7, circle] at (\i-0.5,-1.65) {$\u$};
		}
	\end{tikzpicture}
	\end{minipage}\qquad
	\begin{minipage}{0.4\linewidth}
		\begin{tikzpicture}
	[yscale=0.47,xscale=0.75,voter/.style={anchor=south, yshift=-3pt}, s/.style={fill=green!10}, n/.style={}, c/.style={anchor=south, yshift=2.5pt, inner sep=0}]
		\foreach \i / \x / \y / \wid / \pick in 
			{1/0/0/4/s, 2/0/1/4/s, 
			3/4/0/4/s, 4/4/1/4/s,
			5/0/2/2/s, 6/0/3/2/s, 7/0/4/2/s, 8/0/5/2/s, 9/0/6/2/s,
			10/2/2/2/s, 11/2/3/2/s, 12/2/4/2/s, 13/2/5/2/s, 14/2/6/2/s,
			15/4/2/2/s, 16/4/3/2/s, 17/4/4/2/s, 18/4/5/2/n,
			19/6/2/2/s, 20/6/3/2/s, 21/6/4/2/s, 22/6/5/2/n} {
			\draw[\pick] (\x,\y) rectangle (\x + \wid,\y+1);
			\node[c] at (\x + 0.5 * \wid, \y) {$c_{\i}$}; 	
		}
		\foreach \i/\u in {1/7,2/7,3/7,4/7,5/5,6/5,7/5,8/5} {
			\node[voter] at (\i-0.5,-1) {$v_\i$};
			\node[font=\footnotesize, inner sep=2pt, fill=black!7, circle] at (\i-0.5,-1.65) {$\u$};
		}
	\end{tikzpicture}
	\end{minipage}
\end{center}
All committees that are laminar proportional (such as the blue committee) induce the utility vector $\mathbf w_1 = (6,6,6,6,6,6,6,6)$. The green committee induces $\mathbf w_2 = (7,7,7,7,5,5,5,5)$. Since $f$ satisfies laminar proportionality, $f$ selects only committees with welfare vector $\mathbf w_1$ and none with vector $\mathbf w_2$. Thus, by definition of welfarist rules, we have $g_{20}(\mathbf w_1) > g_{20}(\mathbf w_2)$.

Now consider a second laminar instance as follows; again $k = 20$.
\begin{center}
	\begin{minipage}{0.4\linewidth}
		\begin{tikzpicture}
	[yscale=0.47,xscale=0.75,voter/.style={anchor=south, yshift=-3pt}, s/.style={fill=green!10}, n/.style={}, c/.style={anchor=south, yshift=2.5pt, inner sep=0}]
		\foreach \i / \x / \y / \wid / \pick in 
			{1/0/0/4/s, 2/0/1/4/s, 3/0/2/4/s, 4/0/3/4/s, 5/0/4/4/s, 6/0/5/4/s, 
			7/4/0/2/s, 8/4/1/2/s, 9/4/2/2/s, 10/4/3/2/s, 11/4/4/2/s,
			12/6/0/2/s, 13/6/1/2/s, 14/6/2/2/s, 15/6/3/2/s, 16/6/4/2/s,
			17/0/6/1/n, 18/1/6/1/n, 19/2/6/1/n, 20/3/6/1/n,
			21/4/5/1/s, 22/5/5/1/s, 23/6/5/1/s, 24/7/5/1/s} {
			\draw[\pick] (\x,\y) rectangle (\x + \wid,\y+1);
			\node[c] at (\x + 0.5 * \wid, \y) {$c_{\i}$}; 	
		}
		\foreach \i/\u in {1/6,2/6,3/6,4/6,5/6,6/6,7/6,8/6} {
			\node[voter] at (\i-0.5,-1) {$v_\i$};
			\node[font=\footnotesize, inner sep=2pt, fill=black!7, circle] at (\i-0.5,-1.65) {$\u$};
		}
	\end{tikzpicture}
	\end{minipage}\qquad
	\begin{minipage}{0.4\linewidth}
		\begin{tikzpicture}
	[yscale=0.47,xscale=0.75,voter/.style={anchor=south, yshift=-3pt}, s/.style={fill=blue!10}, n/.style={}, c/.style={anchor=south, yshift=2.5pt, inner sep=0}]
		\foreach \i / \x / \y / \wid / \pick in 
			{1/0/0/4/s, 2/0/1/4/s, 3/0/2/4/s, 4/0/3/4/s, 5/0/4/4/s, 6/0/5/4/s, 
			7/4/0/2/s, 8/4/1/2/s, 9/4/2/2/s, 10/4/3/2/s, 11/4/4/2/s,
			12/6/0/2/s, 13/6/1/2/s, 14/6/2/2/s, 15/6/3/2/s, 16/6/4/2/s,
			17/0/6/1/s, 18/1/6/1/s, 19/2/6/1/s, 20/3/6/1/s,
			21/4/5/1/n, 22/5/5/1/n, 23/6/5/1/n, 24/7/5/1/n} {
			\draw[\pick] (\x,\y) rectangle (\x + \wid,\y+1);
			\node[c] at (\x + 0.5 * \wid, \y) {$c_{\i}$}; 	
		}
		\foreach \i/\u in {1/7,2/7,3/7,4/7,5/5,6/5,7/5,8/5} {
			\node[voter] at (\i-0.5,-1) {$v_\i$};
			\node[font=\footnotesize, inner sep=2pt, fill=black!7, circle] at (\i-0.5,-1.65) {$\u$};
		}
	\end{tikzpicture}
	\end{minipage}
\end{center}
In this instance, the blue committee on the right is the unique laminar proportional committee. It induces welfare vector $\mathbf w_2$. The green committee on the left induces $\mathbf w_1$. Hence, since $g_{20}(\mathbf w_1) > g_{20}(\mathbf w_2)$, the rule $f$ does not select the blue committee at this instance. Hence, $f$ fails laminar proportionality, a contradiction.
\end{proof}

\section{Price Systems}

In this section, we discuss a proportionality axiom we call priceability. We give each voter an equal budget of virtual money, and let voters spend it on candidates they approve. A committee is priceable if there is a price such that voter spending can be arranged in such a way that each committee member gets a total spending of exactly the price, and voters do not have enough money left to buy additional candidates. The intuition behind this condition is that it encodes that each voter has (approximately) equal influence on the committee (since each voter starts out with an equal budget), and this ensures proportionality.

\subsection{Definition and properties}
\label{sec:price}
Assume that each voter has a budget of one dollar, and that she can use this money to pay for candidates that she approves of. A \myemph{price system} is a pair $\textsf{ps} = (p, \{p_i\}_{i \in [n]})$, where $p > 0$ is a \myemph{price}, and for each voter $i \in [n]$, there is a \myemph{payment function} $p_i\colon C \to [0, 1]$ such that
\begin{enumerate}
\item If $p_i(c) > 0$, then $c \in A_i$ (a voter can only pay for candidates she approves of), and
\item $\sum_{c \in \mathbb N} p_i(c) \leq 1$ (a voter can spend at most one dollar) \text{.}
\end{enumerate}
We say that a price system $\textsf{ps} = (p, \{p_i\}_{i \in [n]})$ supports a committee $W$ if the following hold:
\begin{enumerate}
\item For each elected candidate $c \in W$, the sum of the payments to this candidate equals the price: $\sum_{i \in [n]}p_i(c) = p$.
\item No candidate outside of the committee $c \notin W$ gets any payment: $\sum_{i \in [n]}p_i(c) = 0$.
\item There exists no unelected candidate $c \notin W$ whose supporters, in total, have a remaining unspent budget of  more than $p$:\footnote{A strict inequality would seem intuitive here, but we need weak inequality to ensure existence on symmetric instances. Consider an instance with two voters, where $v_1$ approves only $c_1$ and $v_2$ only approves $c_2$. If $k = 1$, no committee is priceable when we require strict inequality.}
\begin{align*}
\sum_{i \in [n]\colon c\in A_i} \left(1 - \sum_{c' \in W}p_i(c')\right) \leq p \text{.}
\end{align*}
\end{enumerate}
We say that a committee $W$ is \myemph{priceable} if there exists a price system $\textsf{ps} = (p, \{p_i\}_{i \in [n]})$ that supports it. Note that if $W$ is supported by a price system with price $p$, then $p \le n/|W|$, since the total spending by voters is $p\cdot |W|$ which is at most the total budget $n$.
Note that the definition of priceability does not refer to a target committee size $k$.

We first show that every priceable committee $W$ satisfies \emph{proportional justified representation} (PJR), for committee size $k = |W|$, an axiom introduced by \citet{pjr17}.\footnote{For a target committee size $k$, \Cref{prop:price-implies-pjr} shows that priceable committees $W$ with $|W| = k$ satisfy PJR for the quota $n/k$. There may be priceable committees $W$ with $|W| < k$ which do not satisfy PJR for the quota $n/k$.}

\begin{proposition}
	\label{prop:price-implies-pjr}
	Suppose that $W$ is a priceable committee. Then $W$ provides proportional justified representation: For every group $S \subseteq N$ of voters such that $|\bigcap_{i\in S} A_i| \ge \ell$ and such that $|S| \ge \ell \cdot n/|W|$, the committee $W$ contains at least $\ell$ candidates from $\bigcup_{i \in S} A_i$.
\end{proposition}
\begin{proof}
	Suppose $W$ is supported by the price system $\textsf{ps} = (p, (p_i)_{i \in N})$, and assume for a contradiction that $|W \cap \bigcup_{i \in S} A_i| < \ell$. Because voters only spend their money on approved candidates, the members of $S$ have together paid the price $p$ for at most $\ell - 1$ candidates, so $\sum_{i\in S} \sum_{c \in C} p_i(c) \le (\ell - 1)\cdot p$. Thus, the total amount of unspent budget held by members of $S$ is
	\[ |S| - \sum_{i\in S} \sum_{c \in C} p_i(c) \ge \ell \cdot \tfrac n {|W|} - (\ell - 1)\cdot p = \ell \cdot (\tfrac n {|W|} - p) + p \ge p, \]
	where the last inequality follows since $p \le \frac n {|W|}$. In fact, the last inequality must be strict, because the total spending $p\cdot |W|$ by all voters is strictly less than the total budget $n$ (because the voters in $S$ have a remaining budget of $\ge p > 0$) and hence $p < \frac n {|W|}$. 
	
	Because $|\bigcap_{i\in S} A_i| \ge \ell$, there exists a candidate $c \in \bigcap_{i\in S} A_i \setminus W$. The supporters of $c$ have strictly more than $p$ dollars left, contradicting the definition of a price system.
\end{proof}

Note that one can check whether a given committee is priceable by solving a simple linear program. In contrast, it is coNP-complete to check whether a given committee provides proportional justified representation \citep{AEHLSS18}.

We now show that Phragm\'en's rule always returns a priceable committee. (Thus, for every profile $P$ and every $k$, there always exists a priceable committee of size $k$.) The proof follows very naturally from the definition; this leads us to view priceability as a defining property of ``Phragm\'en-like'' rules, and we note that the Method of Equal Shares is one of them. A similar proof works for what is known in the literature as ``Phragm\'en's optimal rule''.

\begin{proposition}\label[proposition]{thm:phragmen_and_pricing}
Phragm\'en's rule and Equal Shares always return priceable committees.
\end{proposition}
\begin{proof}
The fact that Equal Shares returns priceable committees follows directly from its definition. In the remaining part of the proof we will focus on Phragm\'en's sequential rule.

First, observe that the definition of priceability does not depend on the initial budget of the voters. One can always rescale their initial budgets, the price, and their payment functions by multiplying them by a constant. 

Consider an election instance $(P, k)$ and let $W$ be the committee returned by Phragm\'en's sequential rule for this instance.  Assume that the rule stopped at time $t$ (recall that $t$ does not have to be an integer). Thus, during the execution of the rule each voter earned $t$ dollars in total. We will now construct a price system that supports $W$, where the initial budgets of the voters equals $t$. We set the price to $\nicefrac{n}{k}$.  For each voter $i$ we construct her payment function as follows; by default we set the payments $p_i(c)$ to 0. Whenever a candidate $c$ is added to the committee by the rule, and whenever the rule charges voter $i$ for adding $c$ to the winning committee, we set  $p_i(c) = x$, where $x$ is the amount deducted from $i$. It is clear that each voter spends at most~$t$ dollars, and that each voter pays only for the candidates that she approves of. It also follow directly from the definition of the rule, that each committee member $c \in W$ will receive a total payment of $\nicefrac{n}{k}$. Finally for each $c' \notin W$ the voters who approve $c'$ cannot have a remaining budget of strictly more than $\nicefrac{n}{k}$ dollars in total; otherwise, the rule would have added $c'$ to the committee at a time strictly earlier than $t$.
\end{proof}

Recall that Equal Shares may return committees consisting of fewer than $k$ candidates. By \Cref{thm:phragmen_and_pricing}, this committee is priceable. If desired, one can find a superset of this committee of size $k$ that is also priceable, for example by filling the remaining seats by running Phragm\'en's rule with initial budgets equal to those at the end of the execution of Equal Shares.

\citet{bri-las-sko:c:apportionment} classified different approval-based committee rules by their behavior on party-list profiles. The committee selection problem on party-list profiles has been studied as the \myemph{apportionment problem}. Both Thiele's and Phragm\'en's rule implement the D'Hondt method of apportionment. For Phragm\'en's rule, this follows from the fact that it is priceable, as we now show.

\begin{theorem}
	\label{thm:price-dhondt}
For each party-list profile $P$, a committee $W$ is priceable if and only if $W$ is selected by the D'Hondt method of apportionment with committee size $k = |W|$.
\end{theorem} 
\begin{proof}
Consider a party list election instance $(P, k)$ with $r$ parties $P_1, P_2, \ldots, P_{r}$; for $z \in [r]$ let $n_z$ denote the number of voters supporting party $P_z$, $\sum_{z \in [r]} n_z = n$.

The fact that each committee returned by the D'Hondt method of apportionment is supported by a price system follows because this committee is selected by Phragm\'en's sequential rule~\cite[Theorem~3]{bri-las-sko:c:apportionment}, whose outputs are supported by price systems (\Cref{thm:phragmen_and_pricing}).

Now, we prove the other implication. Assume that $W$ is a committee supported by a price system $\textsf{ps} = (p, \{p_i\}_{i \in [n]})$, and for the sake of contradiction, assume that $W$ is not a committee returned by the D'Hondt method of apportionment. This, in particular means that there exist two parties, $P_x$ and $P_y$, with the number of  seats assigned respectively equal $k_x$ and $k_y$, such that $\nicefrac{n_x}{k_x+1} > \nicefrac{n_y}{k_y}$ (see, e.g., \citealp{lac-sko:t:approval-thiele}). The total amount of money held by the voters from $P_y$ is equal to $n_y$, and so, the price cannot be higher than $\nicefrac{n_y}{k_y}$. The total amount of unused money, held by the voters from $P_x$ is equal to:
\begin{align*}
n_x - k_x \cdot p \geq n_x - k_x \cdot \frac{n_y}{k_y} > n_x - k_x \cdot \frac{n_x}{k_x + 1} = \frac{n_x}{k_x + 1} > \frac{n_y}{k_y} \geq p\text{.}
\end{align*}
Thus, an additional candidate from $P_x$ can receive the total payment of more than $p$. This contradicts the assumption that $\textsf{ps}$ supports $W$, and completes the proof.
\end{proof}

Equal Shares does not implement the D'Hondt method on party-list. Instead, it gives each party a number of seats that is equal to its \emph{lower quota}. This does not contradict \Cref{thm:price-dhondt} because Equal Shares may return committees of size $|W| < k$. However, all ways of completing Equal Shares to get committee size exactly $k$ will implement the D'Hondt method, as long as these completions preserve priceability.

\subsection{Priceability does not imply efficiency}

Priceability requires only that voters have a similar influence on the outcome. In contrast to other proportionality axioms, it does not imply any guarantees on utility levels of groups. For example, laminar proportionality contains an efficiency component in requiring that unanimous candidates are elected first; this component is missing from priceability. The following example makes this clear, and also shows that priceability and laminar proportionality are logically incomparable.

\begin{example}\label{ex:payments_bad_alloc_1}
Fix $k$, and take $k$ voters and $2k$ candidates. Voter $i$ approves $c_i \cup \{c_{k+1}, \ldots, c_k\}$. The instance is shown below. The committee $\{c_1, \ldots, c_k\}$ (marked blue) is priceable, but clearly, $\{c_{k+1}, \ldots, c_{2k}\}$ (marked green) is a much better choice (and also priceable).

\[
\begin{tikzpicture}
\filldraw[fill=green!10!white, draw=black] (0,0.5) rectangle (4.0,1.0);
\node at (2.0, 0.75) {$c_{k+1}$};
\filldraw[fill=green!10!white, draw=black] (0,1.0) rectangle (4.0,1.5);
\node at (2.0, 1.25) {$c_{k+2}$};
\filldraw[fill=green!10!white, draw=black] (0,1.5) rectangle (4.0,2.0);
\node at (2.0, 1.75) {$\cdots$};
\filldraw[fill=green!10!white, draw=black] (0,2.0) rectangle (4.0,2.5);
\node at (2.0, 2.25) {$c_{2k}$};

\filldraw[fill=blue!10!white, draw=black] (0,0) rectangle (0.8,0.5);
\node at (0.4, 0.25) {$c_1$};
\node at (0.4, -0.35) {$v_1$};
\filldraw[fill=blue!10!white, draw=black] (0.8, 0) rectangle (1.6,0.5);
\node at (1.2, 0.25) {$c_2$};
\node at (1.2, -0.35) {$v_2$};
\filldraw[fill=blue!10!white, draw=black] (1.6, 0) rectangle (2.4,0.5);
\node at (2.0, 0.25) {$c_3$};
\node at (2.0, -0.35) {$v_3$};
\filldraw[fill=blue!10!white, draw=black] (2.4, 0) rectangle (3.2,0.5);
\node at (2.8, 0.25) {$\cdots$};
\node at (2.8, -0.35) {$\cdots$};
\filldraw[fill=blue!10!white, draw=black] (3.2, 0) rectangle (4.0,0.5);
\node at (3.6, 0.25) {$c_{k}$};
\node at (3.6, -0.35) {$v_{k}$};
\end{tikzpicture}
\tag*{\qed}
\]
\end{example} 

The blue committee in the example is not Pareto-efficient, and is arbitrarily bad in terms of utilitarian welfare as $k \to \infty$. Given \Cref{prop:price-implies-pjr}, this example also shows that proportional justified representation is a weak axiom, and does not rule out bad committees. In contrast, extended justified representation (which we will further discuss in \Cref{sec:deviations}) forces the efficient green committee. This is a strong reason to prefer the latter axiom.

\subsection{No welfarist rule is priceable}
Using a similar technique to \Cref{thm:no_welfarist_lam_prop}, we can show a conflict between welfarism and priceability. We show that no Pareto-optimal welfarist rule can satisfy priceability. The additional assumption of Pareto-optimality seems mild in the context of welfarist rules; we conjecture that the incompatibility holds even without assuming Pareto-optimality.

\begin{figure}[t!b]
\begin{center}
Profile 1:
\vspace{0.3cm}

\begin{tikzpicture}
\filldraw[fill=blue!10!white, draw=black] (0,0.0) rectangle (4.8,0.5);
\node at (2.4, 0.25) {$c_{1}$};
\filldraw[fill=blue!10!white, draw=black] (0,0.5) rectangle (4.8,1.0);
\node at (2.4, 0.75) {$c_{2}$};
\filldraw[fill=blue!10!white, draw=black] (0,1.0) rectangle (4.8,1.5);
\node at (2.4, 1.25) {$c_3$};

\filldraw[fill=blue!10!white, draw=black] (0, 1.5) rectangle (0.8,2.0);
\node at (0.4, -0.35) {$v_1$};
\filldraw[fill=blue!10!white, draw=black] (0.8, 1.5) rectangle (1.6,2.0);
\node at (1.2, -0.35) {$v_2$};
\filldraw[fill=blue!10!white, draw=black] (1.6, 1.5) rectangle (2.4,2.0);
\node at (2.0, -0.35) {$v_3$};
\filldraw[fill=blue!10!white, draw=black] (2.4, 1.5) rectangle (3.2,2.0);
\node at (2.8, -0.35) {$v_4$};
\filldraw[fill=blue!10!white, draw=black] (3.2, 1.5) rectangle (4.0,2.0);
\node at (3.6, -0.35) {$v_{5}$};
\filldraw[fill=blue!10!white, draw=black] (4.0, 1.5) rectangle (4.8,2.0);
\node at (4.4, -0.35) {$v_{6}$};

\filldraw[fill=blue!10!white, draw=black] (0, 2.0) rectangle (0.8,2.5);
\filldraw[fill=blue!10!white, draw=black] (0.8, 2.0) rectangle (1.6,2.5);
\filldraw[fill=blue!10!white, draw=black] (1.6, 2.0) rectangle (2.4,2.5);
\filldraw[fill=blue!10!white, draw=black] (2.4, 2.0) rectangle (3.2,2.5);
\filldraw[fill=blue!10!white, draw=black] (3.2, 2.0) rectangle (4.0,2.5);
\filldraw[fill=blue!10!white, draw=black] (4.0, 2.0) rectangle (4.8,2.5);

\filldraw[fill=blue!10!white, draw=black] (0, 2.5) rectangle (0.8,3.0);
\filldraw[fill=blue!10!white, draw=black] (0.8, 2.5) rectangle (1.6,3.0);
\filldraw[fill=blue!10!white, draw=black] (1.6, 2.5) rectangle (2.4,3.0);
\filldraw[fill=blue!10!white, draw=black] (2.4, 2.5) rectangle (3.2,3.0);
\filldraw[fill=blue!10!white, draw=black] (3.2, 2.5) rectangle (4.0,3.0);
\filldraw[fill=blue!10!white, draw=black] (4.0, 2.5) rectangle (4.8,3.0);

\filldraw[fill=blue!10!white, draw=black] (0, 3.0) rectangle (0.8,3.5);
\filldraw[fill=blue!10!white, draw=black] (0.8, 3.0) rectangle (1.6,3.5);
\filldraw[fill=blue!10!white, draw=black] (1.6, 3.0) rectangle (2.4,3.5);
\filldraw[fill=blue!10!white, draw=black] (2.4, 3.0) rectangle (3.2,3.5);
\filldraw[fill=blue!10!white, draw=black] (3.2, 3.0) rectangle (4.0,3.5);
\filldraw[fill=blue!10!white, draw=black] (4.0, 3.0) rectangle (4.8,3.5);

\filldraw[fill=white, draw=black] (0, 3.5) rectangle (0.8,4.0);
\filldraw[fill=white, draw=black] (0.8, 3.5) rectangle (1.6,4.0);
\filldraw[fill=white, draw=black] (1.6, 3.5) rectangle (2.4,4.0);
\filldraw[fill=white, draw=black] (2.4, 3.5) rectangle (3.2,4.0);
\filldraw[fill=white, draw=black] (3.2, 3.5) rectangle (4.0,4.0);
\filldraw[fill=white, draw=black] (4.0, 3.5) rectangle (4.8,4.0);

\node at (2.4, 4.25) {$\cdots$};

\filldraw[fill=white, draw=black] (0, 4.5) rectangle (0.8,5.0);
\filldraw[fill=white, draw=black] (0.8, 4.5) rectangle (1.6,5.0);
\filldraw[fill=white, draw=black] (1.6, 4.5) rectangle (2.4,5.0);
\filldraw[fill=white, draw=black] (2.4, 4.5) rectangle (3.2,5.0);
\filldraw[fill=white, draw=black] (3.2, 4.5) rectangle (4.0,5.0);
\filldraw[fill=white, draw=black] (4.0, 4.5) rectangle (4.8,5.0);

\filldraw[fill=blue!10!white, draw=black] (4.8, 0.0) rectangle (5.6,0.5);
\node at (5.2, -0.35) {$v_7$};
\filldraw[fill=blue!10!white, draw=black] (5.6, 0.0) rectangle (6.4,0.5);
\node at (6.0, -0.35) {$v_8$};
\filldraw[fill=blue!10!white, draw=black] (6.4, 0.0) rectangle (7.2,0.5);
\node at (6.8, -0.35) {$v_9$};
\filldraw[fill=blue!10!white, draw=black] (7.2, 0.0) rectangle (8.0,0.5);
\node at (7.6, -0.35) {$v_{10}$};
\filldraw[fill=blue!10!white, draw=black] (8.0, 0.0) rectangle (8.8,0.5);
\node at (8.4, -0.35) {$v_{11}$};
\filldraw[fill=blue!10!white, draw=black] (8.8, 0.0) rectangle (9.6,0.5);
\node at (9.2, -0.35) {$v_{12}$};

\filldraw[fill=blue!10!white, draw=black] (4.8, 0.5) rectangle (5.6,1.0);
\filldraw[fill=blue!10!white, draw=black] (5.6, 0.5) rectangle (6.4,1.0);
\filldraw[fill=blue!10!white, draw=black] (6.4, 0.5) rectangle (7.2,1.0);
\filldraw[fill=blue!10!white, draw=black] (7.2, 0.5) rectangle (8.0,1.0);
\filldraw[fill=blue!10!white, draw=black] (8.0, 0.5) rectangle (8.8,1.0);
\filldraw[fill=blue!10!white, draw=black] (8.8, 0.5) rectangle (9.6,1.0);

\filldraw[fill=blue!10!white, draw=black] (4.8, 1.0) rectangle (5.6,1.5);
\filldraw[fill=blue!10!white, draw=black] (5.6, 1.0) rectangle (6.4,1.5);
\filldraw[fill=blue!10!white, draw=black] (6.4, 1.0) rectangle (7.2,1.5);
\filldraw[fill=blue!10!white, draw=black] (7.2, 1.0) rectangle (8.0,1.5);
\filldraw[fill=blue!10!white, draw=black] (8.0, 1.0) rectangle (8.8,1.5);
\filldraw[fill=blue!10!white, draw=black] (8.8, 1.0) rectangle (9.6,1.5);

\filldraw[fill=blue!10!white, draw=black] (4.8, 1.5) rectangle (5.6,2.0);
\filldraw[fill=blue!10!white, draw=black] (5.6, 1.5) rectangle (6.4,2.0);
\filldraw[fill=blue!10!white, draw=black] (6.4, 1.5) rectangle (7.2,2.0);
\filldraw[fill=blue!10!white, draw=black] (7.2, 1.5) rectangle (8.0,2.0);
\filldraw[fill=blue!10!white, draw=black] (8.0, 1.5) rectangle (8.8,2.0);
\filldraw[fill=blue!10!white, draw=black] (8.8, 1.5) rectangle (9.6,2.0);

\filldraw[fill=blue!10!white, draw=black] (4.8, 2.0) rectangle (5.6,2.5);
\filldraw[fill=blue!10!white, draw=black] (5.6, 2.0) rectangle (6.4,2.5);
\filldraw[fill=blue!10!white, draw=black] (6.4, 2.0) rectangle (7.2,2.5);
\filldraw[fill=blue!10!white, draw=black] (7.2, 2.0) rectangle (8.0,2.5);
\filldraw[fill=blue!10!white, draw=black] (8.0, 2.0) rectangle (8.8,2.5);
\filldraw[fill=blue!10!white, draw=black] (8.8, 2.0) rectangle (9.6,2.5);

\filldraw[fill=white, draw=black] (4.8, 2.5) rectangle (5.6,3.0);
\filldraw[fill=white, draw=black] (5.6, 2.5) rectangle (6.4,3.0);
\filldraw[fill=white, draw=black] (6.4, 2.5) rectangle (7.2,3.0);
\filldraw[fill=white, draw=black] (7.2, 2.5) rectangle (8.0,3.0);
\filldraw[fill=white, draw=black] (8.0, 2.5) rectangle (8.8,3.0);
\filldraw[fill=white, draw=black] (8.8, 2.5) rectangle (9.6,3.0);

\node at (7.2, 3.25) {$\cdots$};

\filldraw[fill=white, draw=black] (4.8, 3.5) rectangle (5.6,4.0);
\filldraw[fill=white, draw=black] (5.6, 3.5) rectangle (6.4,4.0);
\filldraw[fill=white, draw=black] (6.4, 3.5) rectangle (7.2,4.0);
\filldraw[fill=white, draw=black] (7.2, 3.5) rectangle (8.0,4.0);
\filldraw[fill=white, draw=black] (8.0, 3.5) rectangle (8.8,4.0);
\filldraw[fill=white, draw=black] (8.8, 3.5) rectangle (9.6,4.0);

\end{tikzpicture}
\end{center}

\begin{center}
Profile 2:
\vspace{0.3cm}

\begin{tikzpicture}
\filldraw[fill=blue!10!white, draw=black] (4.8,0.0) rectangle (9.6,0.5);
\node at (7.2, 0.25) {$c_{1}$};
\filldraw[fill=blue!10!white, draw=black] (4.8,0.5) rectangle (9.6,1.0);
\node at (7.2, 0.75) {$c_{2}$};

\filldraw[fill=blue!10!white, draw=black] (4.8,1.0) rectangle (9.6,1.5);
\node at (7.2, 1.25) {$c_3$};

\filldraw[fill=blue!10!white, draw=black] (4.8, 1.5) rectangle (5.6,2.0);
\filldraw[fill=blue!10!white, draw=black] (5.6, 1.5) rectangle (6.4,2.0);
\filldraw[fill=blue!10!white, draw=black] (6.4, 1.5) rectangle (7.2,2.0);
\filldraw[fill=blue!10!white, draw=black] (7.2, 1.5) rectangle (8.0,2.0);
\filldraw[fill=blue!10!white, draw=black] (8.0, 1.5) rectangle (8.8,2.0);
\filldraw[fill=blue!10!white, draw=black] (8.8, 1.5) rectangle (9.6,2.0);

\filldraw[fill=blue!10!white, draw=black] (4.8, 2.0) rectangle (5.6,3.0);
\filldraw[fill=blue!10!white, draw=black] (5.6, 2.0) rectangle (6.4,3.0);
\filldraw[fill=blue!10!white, draw=black] (6.4, 2.0) rectangle (7.2,3.0);
\filldraw[fill=blue!10!white, draw=black] (7.2, 2.0) rectangle (8.0,3.0);
\filldraw[fill=blue!10!white, draw=black] (8.0, 2.0) rectangle (8.8,3.0);
\filldraw[fill=blue!10!white, draw=black] (8.8, 2.0) rectangle (9.6,3.0);

\filldraw[fill=blue!10!white, draw=black] (4.8, 2.5) rectangle (5.6,3.0);
\filldraw[fill=blue!10!white, draw=black] (5.6, 2.5) rectangle (6.4,3.0);
\filldraw[fill=blue!10!white, draw=black] (6.4, 2.5) rectangle (7.2,3.0);
\filldraw[fill=blue!10!white, draw=black] (7.2, 2.5) rectangle (8.0,3.0);
\filldraw[fill=blue!10!white, draw=black] (8.0, 2.5) rectangle (8.8,3.0);
\filldraw[fill=blue!10!white, draw=black] (8.8, 2.5) rectangle (9.6,3.0);

\filldraw[fill=blue!10!white, draw=black] (4.8, 3.0) rectangle (5.6,3.5);
\filldraw[fill=blue!10!white, draw=black] (5.6, 3.0) rectangle (6.4,3.5);
\filldraw[fill=blue!10!white, draw=black] (6.4, 3.0) rectangle (7.2,3.5);
\filldraw[fill=blue!10!white, draw=black] (7.2, 3.0) rectangle (8.0,3.5);
\filldraw[fill=blue!10!white, draw=black] (8.0, 3.0) rectangle (8.8,3.5);
\filldraw[fill=blue!10!white, draw=black] (8.8, 3.0) rectangle (9.6,3.5);

\filldraw[fill=white, draw=black] (4.8, 3.5) rectangle (5.6,4.0);
\filldraw[fill=white, draw=black] (5.6, 3.5) rectangle (6.4,4.0);
\filldraw[fill=white, draw=black] (6.4, 3.5) rectangle (7.2,4.0);
\filldraw[fill=white, draw=black] (7.2, 3.5) rectangle (8.0,4.0);
\filldraw[fill=white, draw=black] (8.0, 3.5) rectangle (8.8,4.0);
\filldraw[fill=white, draw=black] (8.8, 3.5) rectangle (9.6,4.0);

\node at (7.2, 4.25) {$\cdots$};

\filldraw[fill=white, draw=black] (4.8, 4.5) rectangle (5.6,5.0);
\filldraw[fill=white, draw=black] (5.6, 4.5) rectangle (6.4,5.0);
\filldraw[fill=white, draw=black] (6.4, 4.5) rectangle (7.2,5.0);
\filldraw[fill=white, draw=black] (7.2, 4.5) rectangle (8.0,5.0);
\filldraw[fill=white, draw=black] (8.0, 4.5) rectangle (8.8,5.0);
\filldraw[fill=white, draw=black] (8.8, 4.5) rectangle (9.6,5.0);

\filldraw[fill=blue!10!white, draw=black] (0, 2.0) rectangle (0.8,2.5);
\filldraw[fill=blue!10!white, draw=black] (0.8, 2.0) rectangle (1.6,2.5);
\filldraw[fill=blue!10!white, draw=black] (1.6, 2.0) rectangle (2.4,2.5);
\filldraw[fill=blue!10!white, draw=black] (2.4, 2.0) rectangle (3.2,2.5);
\filldraw[fill=blue!10!white, draw=black] (3.2, 2.0) rectangle (4.0,2.5);
\filldraw[fill=blue!10!white, draw=black] (4.0, 2.0) rectangle (4.8,2.5);

\filldraw[fill=blue!10!white, draw=black] (0, 1.5) rectangle (0.8,2.0);
\filldraw[fill=blue!10!white, draw=black] (0.8, 1.5) rectangle (1.6,2.0);
\filldraw[fill=blue!10!white, draw=black] (1.6, 1.5) rectangle (2.4,2.0);
\filldraw[fill=blue!10!white, draw=black] (2.4, 1.5) rectangle (3.2,2.0);
\filldraw[fill=blue!10!white, draw=black] (3.2, 1.5) rectangle (4.0,2.0);
\filldraw[fill=blue!10!white, draw=black] (4.0, 1.5) rectangle (4.8,2.0);

\filldraw[fill=blue!10!white, draw=black] (0, 0.0) rectangle (0.8,0.5);
\filldraw[fill=blue!10!white, draw=black] (0.8, 0.0) rectangle (1.6,0.5);
\filldraw[fill=blue!10!white, draw=black] (1.6, 0.0) rectangle (2.4,0.5);
\filldraw[fill=blue!10!white, draw=black] (2.4, 0.0) rectangle (3.2,0.5);
\filldraw[fill=blue!10!white, draw=black] (3.2, 0.0) rectangle (4.0,0.5);
\filldraw[fill=blue!10!white, draw=black] (4.0, 0.0) rectangle (4.8,0.5);

\filldraw[fill=blue!10!white, draw=black] (0, 0.5) rectangle (0.8,1.0);
\filldraw[fill=blue!10!white, draw=black] (0.8, 0.5) rectangle (1.6,1.0);
\filldraw[fill=blue!10!white, draw=black] (1.6, 0.5) rectangle (2.4,1.0);
\filldraw[fill=blue!10!white, draw=black] (2.4, 0.5) rectangle (3.2,1.0);
\filldraw[fill=blue!10!white, draw=black] (3.2, 0.5) rectangle (4.0,1.0);
\filldraw[fill=blue!10!white, draw=black] (4.0, 0.5) rectangle (4.8,1.0);

\filldraw[fill=blue!10!white, draw=black] (0, 1.0) rectangle (0.8,1.5);
\filldraw[fill=blue!10!white, draw=black] (0.8, 1.0) rectangle (1.6,1.5);
\filldraw[fill=blue!10!white, draw=black] (1.6, 1.0) rectangle (2.4,1.5);
\filldraw[fill=blue!10!white, draw=black] (2.4, 1.0) rectangle (3.2,1.5);
\filldraw[fill=blue!10!white, draw=black] (3.2, 1.0) rectangle (4.0,1.5);
\filldraw[fill=blue!10!white, draw=black] (4.0, 1.0) rectangle (4.8,1.5);

\filldraw[fill=white, draw=black] (0, 2.5) rectangle (0.8,3.0);
\filldraw[fill=white, draw=black] (0.8, 2.5) rectangle (1.6,3.0);
\filldraw[fill=white, draw=black] (1.6, 2.5) rectangle (2.4,3.0);
\filldraw[fill=white, draw=black] (2.4, 2.5) rectangle (3.2,3.0);
\filldraw[fill=white, draw=black] (3.2, 2.5) rectangle (4.0,3.0);
\filldraw[fill=white, draw=black] (4.0, 2.5) rectangle (4.8,3.0);

\node at (2.4, 3.25) {$\cdots$};

\filldraw[fill=white, draw=black] (0, 3.5) rectangle (0.8,4.0);
\filldraw[fill=white, draw=black] (0.8, 3.5) rectangle (1.6,4.0);
\filldraw[fill=white, draw=black] (1.6, 3.5) rectangle (2.4,4.0);
\filldraw[fill=white, draw=black] (2.4, 3.5) rectangle (3.2,4.0);
\filldraw[fill=white, draw=black] (3.2, 3.5) rectangle (4.0,4.0);
\filldraw[fill=white, draw=black] (4.0, 3.5) rectangle (4.8,4.0);

\node at (0.4, -0.35) {$v_1$};
\node at (1.2, -0.35) {$v_2$};
\node at (2.0, -0.35) {$v_3$};
\node at (2.8, -0.35) {$v_4$};
\node at (3.6, -0.35) {$v_{5}$};
\node at (4.4, -0.35) {$v_{6}$};
\node at (5.2, -0.35) {$v_7$};
\node at (6.0, -0.35) {$v_8$};
\node at (6.8, -0.35) {$v_9$};
\node at (7.6, -0.35) {$v_{10}$};
\node at (8.4, -0.35) {$v_{11}$};
\node at (9.2, -0.35) {$v_{12}$};

\end{tikzpicture}
\end{center}
\caption{A diagram illustrating profiles used in the proof that the requirement of being supported by a price system is incompatible with welfarism (subject to Pareto-optimality). Each of the first six voters in each of the two profiles approves at least 57 candidates. }
\label{fig:welfarism_and_pricing}
\end{figure}

\begin{theorem}
There exists no Pareto-optimal welfarist rule that would always return committees supported by price systems.
\end{theorem}
\begin{proof}
For the sake of contradiction assume that there exists a Pareto-optimal welfarist rule $f$ that always returns committees supported by price systems. 

Consider Profile 1 depicted in \Cref{fig:welfarism_and_pricing}. There, we have 12 voters and the committee size is $k = 57$; each voter approves at least $57$ candidates. Let $W\subseteq C$ with $|W| = k$ be a committee supported by some price system $\textsf{ps} = (p, \{p_i\}_{i \in [n]})$.

We first prove that each of the last 6 voters has at most 5 representatives in $W$. Assume towards a contradiction that one of the last 6 voters, say $v_7$, has at least 6 representatives. This means that $p \le \nicefrac{1}{6}$, since no other voter approves candidates approved by $v_7$, and thus $v_7$ must fully pay for these. Since $p \le \nicefrac{1}{6}$, each of the last 6 voters must have at least 5 representatives, since they otherwise have a remaining budget of strictly more than $p$. Thus, at least 31 seats are filled by candidates from $A_7 \cup \cdots \cup A_{12}$, leaving at most $57 - 31 = 26$ seats for other candidates. Hence, the total spending of voters $v_1,\dots,v_6$ is at most $\frac{26}{6}$. Thus, the average remaining budget among $v_1,\dots,v_6$ is at least
\begin{align*}
\frac{6 - \frac{26}{6}}{6} = 1 -  \frac{13}{18} = \frac{5}{18} > \frac{1}{6}.
\end{align*}
Hence, one of these voters has a remaining budget of strictly more than $p$, and could afford an additional representative, a contradiction.

Second, we prove that each of the last 6 voters has at least 4 representatives in the winning committee. For a contradiction, assume that one voter, say $v_7$, has at most 3 representatives. This means that $p \ge \nicefrac{1}{4}$. Thus, each of the last 6 voters can have at most 4 representatives in the committee. So there are at most $5 \cdot 4 + 3 = 23$ candidates from $A_7 \cup \cdots \cup A_{12}$ in $W$. This leaves $57 - 23 = 34$ seats for the remaining candidates, and thus $W$ contains at least $31$ candidates from $C_\dagger = (A_1 \cup \cdots \cup A_6) \setminus \{c_1, \dots, c_3\}$. Thus, one of $v_1, \dots, v_6$, say $v_1$ has at least 5 approved candidates in $W \cap C_\dagger$.  Since no other voter approves those candidates, we have $p \le \nicefrac{1}{5}$, a contradiction. 

Summarizing, each of the last 6 voters has 4 or 5 representatives in $W$, and so $W$ contains at least $57-30 = 27$ candidates from $A_1 \cup \cdots \cup A_6$. Now let $W$ be a committee selected by $f$. Since $f$ is Pareto-optimal, $W$ contains candidates $c_1$, $c_2$, and $c_3$. Then $W$ contains at least 24 candidates from $C_\dagger$, and so $v_1, \dots, v_6$ have on average at least $7$ representatives in $W$. Let $\mathbf w_1$ be the welfare vector corresponding to the winning committee.

Consider Profile 2, which is obtained from Profile 1 by exchanging the roles of the first and the last 6 voters. By the same reasoning as before, we infer that in the winning committee each of the first 6 voters has 4 or 5 representatives, that $c_1$, $c_2$, and $c_3$ are selected, and that the last 6 voters have on average at least 7 representatives. Let $\mathbf w_2$ denote the welfare vector corresponding to the winning committee in Profile 2.

Observe that both $\mathbf w_1$ and $\mathbf w_2$ are achievable in Profile 1 and in Profile 2. Thus, based on the analysis of Profile 1 we get that $f$ prefers $\mathbf w_1$ over $\mathbf w_2$, and based on the analysis of Profile 2 we get that $f$ prefers $\mathbf w_2$ over $\mathbf w_1$, a contradiction.
\end{proof}

\section{The Core}

The main stability notion of cooperative game theory, the core, can be applied to the committee setting. Consider an election instance $(P,k)$. Based on a definition of \citet{justifiedRepresenattion}, we say that a committee $W$ with $|W| \le k$ is in the \myemph{core} if there does not exist a group of voters $S \subseteq N$ and a set of candidates $T$ such that $\frac{|T|}{k} \le \frac{|S|}{n}$ and for each voter $i \in S$ we have $|A_i \cap T| > |A_i \cap W|$. For example, a committee would be unstable if there existed a coalition of $40\%$ of the voters who can propose a committee that has $40\%$ of the allowed size such that each coalition member strictly prefers the smaller committee they proposed. The core implies other axioms such as extended justified representation (EJR). All known committee rules fail the core, and it is not known whether it is always non-empty.\footnote{If one generalizes from approval preferences to arbitrary additive utilities, there are small instances where the core is empty \citep{FMS18}.}

The open question of core existence has proved tantalizing, and many researchers have tried to settle the problem, leading to positive results when allowing randomization or approximation. We survey this work at the end of Section~\ref{sec:approx}. Simultaneously, due to its intuitive appeal and simple definition, some researchers have identified the core as the ``ultimate'' notion of proportionality and fairness in the committee setting. While we agree that the core idea is attractive, it is worth revisiting our earlier example which shows that if we accept the core, we must reject some egalitarian intuitions.

An outcome admits a \myemph{Pigou--Dalton transfer} if we can obtain an alternative outcome where the utility of all but two individuals remains unchanged, and we have reduced the inequality between those two individuals: the individual with higher utility was lowered, and the individual with lower utility was raised (without changing their relative order nor the sum of their utilities). The \myemph{Pigou--Dalton principle} prescribes that we should only select outcomes that do not admit a Pigou--Dalton transfer. This is a minimal condition for preferring more equal outcomes.

\begin{figure}[t!b]
\centering
\begin{tikzpicture}
[yscale=0.43,xscale=0.78,voter/.style={anchor=south, yshift=-7pt}, select/.style={fill=blue!10}, c/.style={anchor=south, yshift=1.5pt, inner sep=0}]
	\draw[select] (0,0) rectangle (3,1);
	\draw[select] (0,1) rectangle (3,2);
	\draw[select] (0,2) rectangle (3,3);
	\draw[select] (0,3) rectangle (1,4);
	\draw[select] (1,3) rectangle (2,4);
	\draw[select] (2,3) rectangle (3,4);
	\node at (1.5,0.42) {$c_1$};
	\node at (1.5,1.42) {$c_2$};
	\node at (1.5,2.42) {$c_3$};
	\node at (0.5,3.42) {$c_4$};
	\node at (1.5,3.42) {$c_5$};
	\node at (2.5,3.42) {$c_6$};
	\foreach \x in {3,4,5}
		{
		\foreach \y in {0,1}
			{
			\draw[select] (\x,\y) rectangle (\x+1,\y+1);
			}
		\foreach \y in {2}
			{
			\draw (\x,\y) rectangle (\x+1,\y+1);
			}
		}
	\node at (3.5,0.42) {$c_{7}$};
	\node at (3.5,1.42) {$c_{8}$};
	\node at (3.5,2.42) {$c_{9}$};
	\node at (4.5,0.42) {$c_{10}$};
	\node at (4.5,1.42) {$c_{11}$};
	\node at (4.5,2.42) {$c_{12}$};
	\node at (5.5,0.42) {$c_{13}$};
	\node at (5.5,1.42) {$c_{14}$};
	\node at (5.5,2.42) {$c_{15}$};
	\foreach \i in {1,...,6}
		\node[voter] at (\i-0.5,-1) {$v_\i$};
\end{tikzpicture}
\caption{An election instance with $k=12$ in which every core outcome violates the Pigou--Dalton principle.}
\label{fig:core_pigou_dalton}
\end{figure}
Consider the election instance shown in \Cref{fig:core_pigou_dalton} for $k = 12$, which we previously considered in the introduction. Any committee $W$, to be in the core, must contain $c_1$, $c_2$, $c_3$, and at least one of $c_4$, $c_5$, $c_6$. Otherwise, the coalition $\{v_1, v_2, v_3\}$ can block $W$ by proposing $\{c_1, \dots, c_6\}$. This leaves at most $k - 4 = 8$ committee seats to be filled with candidates approved by voters $v_4, v_5, v_6$. Hence, in $W$, there is a voter with utility 4 and a voter with utility 2. Thus, the committee $W$ admits a Pigou--Dalton transfer: take out $c_4$, $c_5$, or $c_6$, and add a candidate approved by a voter with utility 2. This example shows that the core is incompatible with the Pigou--Dalton principle. Any core outcome has avoidable inequality.

The intuitive reason why the core forces inequality is that the core guarantees higher utility to groups that are easily pleased: in our example, the coalition $\{v_1, v_2, v_3\}$ approves three candidates in common, which allows them to demand higher utility than the other voters who are less cohesive.

It is easy to see that PAV satisfies the Pigou--Dalton principle, and thus it is immediate that it fails the core. In the next section, we will quantify the incompatibility between Pigou--Dalton and the core, and we will see that PAV comes as close as possible to satisfying the core, subject to honoring the Pigou--Dalton principle.

\subsection{Approximations to the core}
\label{sec:approx}
For $\lambda \ge 1$, we say that a committee $W$ is in the \myemph{$\lambda$-core} if there does not exist a coalition of voters $S \subseteq N$ and a set of candidates $T$ such that $\frac{|T|}{k} \le \frac{|S|}{n}$ and for each voter $i \in S$ we have $|A_i \cap T| > \max(\lambda \cdot |A_i \cap W|, 1)$. Lower values of $\lambda$ correspond to stronger stability guarantees. For example, for $\lambda = 2$, the condition requires that no coalition can deviate in such a way that each coalition member doubles their utility. In the definition we use $\max(\lambda \cdot |A_i \cap W|, 1)$ instead of simply $\lambda \cdot |A_i \cap W|$ for a technical reason. If $|A_i \cap W| = 0$ then all, even very high, values of $\lambda$ impose the same (weak) constraint on the voter's $i$ satisfaction. Our proof for the Method of Equal Shares (\Cref{thm:rule_x_core}) depends on the requirement that $|A_i \cap T| > 1$ whenever $|A_i \cap W| = 0$.\footnote{This is mainly the artifact of the fact that Equal Shares can select less than $k$ candidates. We anticipate that even for the proof of  \Cref{thm:rule_x_core} we could omit the ``max'' operator from the definition of the $\lambda$-core, if we assumed that the committee returned by Equal Shares is enlarged in some specific way.} All the remaining results would hold even if we omitted the ``max'' operator from the definition. 

In the previous section, we saw that no rule that satisfies the Pigou--Dalton principle can be in the core. By generalising the example, we can show that any such rule must even fail the $(2-\epsilon)$-core for every $\epsilon > 0$.

\begin{theorem}
\label{thm:2-core-pigou-dalton}
For each $\epsilon > 0$ there exists no rule that satisfies the Pigou--Dalton principle and the $(2-\epsilon)$-core property.
\end{theorem}
\begin{proof}
Let $f$ be a rule that satisfies the $(2-\epsilon)$-core property for some fixed value of $\epsilon > 0$. 

Consider a profile constructed as follows. Let us fix two positive integers, $x$ and $y$, with $x \geq y^2$ and $y > \nicefrac{1}{\epsilon}$. First, we have $x$ voters who approve some $y$ common candidates; additionally, each such a voter approves some $y$ different candidates. Next, we have $yx$ voters, each approving $y$ candidates, different for each voter. This gives in total $m = y + yx + y^2x$ candidates. The committee size is $k = y^2x + y$. First, we show that, according to the core property, the first group of $x$ voters should be able to decide about at least $y + x(y-1)$ candidates. Indeed:
\begin{align*}
\frac{x}{n} \cdot k &= (y^2x + y) \cdot \frac{x}{x + yx} =  \frac{y^2x + y}{1 + y} =  \frac{y^2x +yx - yx - y^2 + y^2 + y}{1 + y}  \\
                            &\geq \frac{y^2x +yx - yx - x + y^2 + y}{1 + y} = yx -x + y = y + (y-1)x \text{.}
\end{align*}
Thus, this group of voters can suggest a subset of candidates where each voter has $2y-1$ representatives. 
Since $f$ satisfies the $(2-\epsilon)$-core property, we infer that at least one of the first $x$ voters must have at least $y+1$ representatives; indeed this is because $\frac{2y-1}{y} > 2 - \epsilon$. Thus, among the last $yx$ voters there must exists one that has at most $y-1$ representatives:
\begin{align*}
 \frac{k - y -1}{yx}  < \frac{k - y}{yx} = \frac{y^2x}{yx} = y \text{.}
\end{align*} 
This violates the Pigou--Dalton principle.
\end{proof}

Since PAV satisfies the Pigou--Dalton principle, \Cref{thm:2-core-pigou-dalton} shows that PAV cannot satisfy the $(2-\epsilon)$-core. However, as we now show, it does satisfy the $2$-core. Thus, on this measure, PAV optimally approximates the core subject to Pigou--Dalton.

\begin{theorem}
PAV satisfies the $2$-core property.
\end{theorem}
\begin{proof}
Let $W$ be a committee returned by PAV for some approval profile $P$.
Assume for a contradiction that there is a group of voters $S \subseteq N$ and a set of candidates $T$ such that:
\begin{enumerate}
\item $|T| \leq k \cdot \frac{|S|}{n}$, and
\item For each voter $i \in S$ we have $|A_i \cap T| \geq 2 |A_i \cap W| + 1$.
\end{enumerate} 
For each candidate $c \in T \setminus W$, we write $\Delta_c$ for the increase of the PAV score due to adding $c$ to~$W$. Let us compute the following sum: $\sum_{c \in T \setminus W} \Delta_c$. Consider a voter $i \in S$. Let $r_i = |A_i \cap W|$ denote the number of representatives that $i$ has in committee $W$. Observe that for a candidate $c \in T \setminus W$ who is approved by $i$, by adding $c$ to $W$ we would increase the PAV score that $i$ assigns to the committee by $\frac{1}{r_i + 1}$. Further, observe that the number of candidates from $T \setminus W$ that $i$ approves is at least equal to:
\begin{align}\label{eq:voter_approvals_num}
|A_i \cap T| - |A_i \cap W| \geq 2|A_i \cap W| + 1 - |A_i \cap W| = |A_i \cap W| + 1 = r_i + 1 \text{.}
\end{align}
Thus, we get that:
\begin{align}\label{eq:pav_core_sum_of_mariginals}
\begin{split}
\sum_{c \in T \setminus S} \Delta_c &\geq \sum_{i \in S} \left(|A_i \cap T| - |A_i \cap W|\right) \cdot \frac{1}{r_i + 1} \\
                                    &\geq \sum_{i \in S} (r_i + 1) \cdot \frac{1}{r_i + 1} = \sum_{i \in S} 1 = |S| \text{.}
\end{split}
\end{align}
Consequently, there exists $c \in T \setminus W$, such that
\begin{align}\label{eq:pav_core_pigeonhole}
\Delta_c \geq \frac{|S|}{|T \setminus W|} \geq \frac{|S|}{k \cdot \frac{|S|}{n}} = \frac{n}{k} \text{.}
\end{align}
In fact, the above inequality is strict or at least one voter has no representatives in the committee~$W$. To see that, consider the following two cases. If $T \cap W \neq \emptyset$, then $|T \setminus W| < k \cdot \frac{|S|}{n}$, and \eqref{eq:pav_core_pigeonhole} becomes a strict inequality. Otherwise, if $T \cap W = \emptyset$, then we can observe that the number of members from $T \setminus W$ who are approved by a voter $i$ is at least equal to the number of members from $T$, approved by a voter $i$:
\begin{align*}
|A_i \cap T| \geq 2|A_i \cap W| + 1 = 2r_i + 1\text{.}
\end{align*}
If $r_i \neq 0$ for at least one voter from $S$, then $|A_i \cap T| >  r_i + 1$, and so, when deriving \eqref{eq:pav_core_sum_of_mariginals} we could use \eqref{eq:voter_approvals_num} as a strict inequality, getting \eqref{eq:pav_core_sum_of_mariginals} as a strict inequality. If $r_i = 0$, then, by definition, $i$ has no representatives in the elected committee.

In any case, we have that $\Delta_c > \nicefrac{n}{k}$, or that $\Delta_c \geq \nicefrac{n}{k}$ and $A_i \cap W = \emptyset$ for some $i \in S$.

On the other hand, for $c\in W$ let $\delta_c$ denote the decrease of the PAV score due to removing $c$ from $W$. For each voter $i \in N$ the score that $i$ assigns to $W$ will decrease for exactly $r_i$ members of $W$. Each time, the decrease will be equal to $\frac{1}{r_i}$. Thus, we have that:
\begin{align*}
\sum_{c \in W} \delta_c = \sum_{i\colon r_i > 0} r_i \cdot \frac{1}{r_i} = \sum_{i\colon r_i > 0}1 \leq n \text{.}
\end{align*}
Further, if $A_i \cap S = \emptyset$ for at least one voter, then the above inequality is strict. Consequently, there exists a candidate $c'$ whose removal would decrease the PAV score of $W$ by at most $\nicefrac{n}{k}$ (or by strictly less than $\nicefrac{n}{k}$ if $A_i \cap W = \emptyset$ for some $i \in S$). 

Thus, by replacing $c'$ with $c$ in $W$ we would obtain a committee with a higher PAV score than~$W$, a contradiction. 
\end{proof}

We can also show that Equal Shares approximates the core, although only to a logarithmic factor. The proof is in the appendix.

\newcommand{\ruleXandCoreApproximation}{\mbox{}Equal Shares satisfies the $O\left(\log k\right)$-core property, but fails the  $O\left(\log^c k\right)$-core property for each $c < 1$.}

\begin{theorem}\label{thm:rule_x_core}
\ruleXandCoreApproximation
\end{theorem}

\paragraph{Other approximate core notions.}
Other authors have considered alternative notions of approximate core outcomes.
\citet{FMS18} consider an approximate core notion that uses both multiplicative and additive factors. They show that a rule based on dependent rounding of a fractional committee can approximate the core in their sense. For each $\lambda \in (1,2]$ they show that there exists a committee such that for no coalition $S$ and set $T$ with $|T| \le k \cdot (2 - \lambda) \cdot \frac{|S|}{n}$ we have $|A_i \cap T| \ge \lambda \cdot |A_i \cap W| + \eta$ for all $i \in S$, where $\eta$ is a constant depending only on $\lambda$ and~$k$. Note the extra factor $(2 - \lambda)$ in the upper bound for $|T|$. This notion of approximation is incomparable to ours (for $\lambda = 2$ it is weaker). \citet{cheng2019group} show that under PAV, it can occur that coalitions can deviate using a set $T$ that is much smaller than the limit $k \cdot \frac{|S|}{n}$ imposed by the core; they show that there can be successful deviations with $|T| = O(\beta \cdot \frac{|S|}{n})$ where $\beta = \sqrt{k}$. \citet{jiang2019approx} show that for $\beta = 32$ there always exists an outcome where the deviation is not possible.

\subsection{The core with constrained deviations}
\label{sec:deviations}

In \Cref{thm:rule_x_core}, we saw that Equal Shares only provides a logarithmic approximation to the core. In comparison to PAV, this may sound like bad news. However, the lower bound example for \Cref{thm:rule_x_core} (in the appendix) has an interesting structure: the blocking coalitions in the example profile are very unequal. While some coalition members are extremely satisfied by the blocking proposal, others are barely better off. In this section, we show that whenever Equal Shares fails the core, the output committee can only be blocked by proposals that are ``unfair'' to members of the blocking coalition. In practice, it will be difficult to form such coalitions since they face internal instability. This is a similar principle as the bargaining set from cooperative game theory.

We begin by defining a general notion of the core with constraints on the allowed deviations. This yields a way of studying relaxations to the core different from the quantitative approximation notions discussed in the previous section.

A \myemph{property of committees} $\mathcal P$ is a set of pairs $((P,k), W)$ where $(P,k)$ is an election instance and $W$ is a committee. For example, priceability is a property of committees.

\begin{definition}\label{def:core_relax}
Let $\mathcal{P}$ be a property of committees. Given an instance $(P,k)$ and a committee $W$, we say that a pair $(S, T)$, with $S \subseteq N(P)$ and $T \subseteq C(P)$  is an \myemph{allowed deviation} from $W$ if
\begin{enumerate}[(i)]
\item $|T| \leq k \cdot \nicefrac{|S|}{n}$,
\item for each $i \in S$ we have that $|A_i \cap T| > |A_i \cap W|$, and
\item $T$ has property $\mathcal P$ in the instance $(P|_S, |T|)$ obtained by restricting $P$ to the voters in $S$.
\end{enumerate}
We say that a committee rule $f$ satisfies the \myemph{core subject to $\mathcal{P}$} if for each profile $P$ and each winning committee $W \in f(P)$ there exists no allowed deviation.
\end{definition}

Consider the property $\mathcal{P}_{\mathrm{coh}}$ called \myemph{cohesiveness}, where $((P,k), T) \in \mathcal{P}_{\mathrm{coh}}$ if every voter in $P$ approves all candidates in $T$. It is easy to see that a rule $f$ satisfies the core subject to $\mathcal{P}_{\mathrm{coh}}$ if and only if $f$ satisfies \myemph{extended justified representation} (EJR) \citep{justifiedRepresenattion}. A committee $W$ satisfies EJR if for every group of voters $S\subseteq N$ with $|S| \ge \ell \cdot \nicefrac n k$ and $|\bigcap_{i\in S}A_i| \ge \ell$, there exists some $i\in S$ such that $|W \cap A_i| \ge \ell$.

\begin{theorem}\label{thm:rulex_and_ejr}
The Method of Equal Shares satisfies EJR.
\end{theorem}
\begin{proof}
Let $W$ be a committee picked by Equal Shares for an approval profile $A$. Recall that $p = \nicefrac{n}{k}$ is the price used in the definition of Equal Shares.
Consider an $\ell$-cohesive group of voters $S$ and, for the sake of contradiction, assume that for each $i \in S$ we have $|A_i \cap W| \leq \ell-1$. When the rule stopped, some voter $i \in S$ must have been left with a budget of less than $\nicefrac{p}{|S|}$. Indeed, otherwise there would exist a candidate $c \notin W$ approved by all the voters from $S$ who would be $\nicefrac{p}{|S|}$-affordable, and thus Equal Shares would not have stopped and instead picked $c$.
Since $i$ has at most $\ell-1$ representatives in the committee, for some committee member $i$ must have paid more than:
\begin{align*}
\frac{1 - \frac{p}{|S|}}{\ell-1} \geq \frac{1 - \frac{1}{\ell}}{\ell-1} = \frac{1}{\ell}.
\end{align*}
Let us consider the first committee member $c'$ that has been picked by the rule such that some voter from $S$ paid more than $\nicefrac{1}{\ell}$ for $c'$ (by the above reasoning such a candidate exists). Thus, $c'$ is not $\nicefrac{1}{\ell}$-affordable.
At that moment each voter $i \in S$ had at most $\ell-1$ representatives, and for each of them paid at most $\nicefrac{1}{\ell}$. Thus, each such a voter would have at least the following budget left:
\begin{align*}
1 - (\ell -1) \cdot \frac{1}{\ell} = 1 - 1 + \frac{1}{\ell} = \frac{1}{\ell}.
\end{align*}
Further, $|S| \cdot \nicefrac{1}{\ell} \geq p$. Thus, there is a candidate $c''$ who is approved by all the voters from $S$ and which is $\nicefrac{1}{\ell}$-affordable. This contradicts that Equal Shares chose candidate $c'$ and completes the proof.
\end{proof}

Note that Phragm\'en's rule fails EJR \citep{aaai/BrillFJL17-phragmen}, while PAV satisfies it \citep{justifiedRepresenattion}. Because PAV is NP-hard to compute, \citet{AEHLSS18} looked for polynomial-time computable rules satisfying EJR, and found a class of rather artificial rules doing so. The Method of Equal Shares is arguably the most natural polynomial-time computable committee rule satisfying EJR known to date.

We have argued above that \Cref{def:core_relax} is particularly well motivated when $\mathcal{P}$ is itself a fairness property. \Cref{thm:rulex_relaxed_core} below shows that Equal Shares satisfies the core subject to $\mathcal{P}_{\mathrm{price\text{-}eq}}$, where $\mathcal{P}_{\mathrm{price\text{-}eq}}$ is is a strengthening of the priceability requirement: We say that $((P, k), T) \in \mathcal{P}_{\mathrm{price\text{-}eq}}$ if there exists a family of payment functions $\{p_i\}_{i \in N(P)}$ with
\begin{enumerate}[(i)]
\item $\sum_{c \in T}p_i(c) \leq 1$ for each $i \in S$,
\item $\sum_{i \in N}p_i(c) = \nicefrac{n}{k}$ for each $c \in T$, and
\item for each $i, j \in N(c)$ we have $p_i(c) = p_{j}(c)$.
\end{enumerate}
We call the property $\mathcal{P}_{\mathrm{price\text{-}eq}}$ \myemph{priceability with equal payments}. Note that this property is weaker than cohesiveness (thus, in particular \Cref{thm:rulex_relaxed_core} implies \Cref{thm:rulex_and_ejr}). In fact, it is even weaker than laminar proportionality. It implements an intuitive idea of fairness, which says that different voters should all pay the same amount to elect an approved candidate to the committee. Notably, one can check that in the example discussed in the introduction, PAV admits a core deviation that is priceable with equal payments, so Equal Shares outperforms PAV on this measure.

\begin{theorem}\label{thm:rulex_relaxed_core}
Equal Shares satisfies the core subject to $\mathcal{P}_{\mathrm{price\text{-}eq}}$.
\end{theorem}  
\begin{proof}
Towards a contradiction assume that there exists a profile $P$, a committee $W$ returned by Equal Shares for $P$, and an allowed deviation by coalition $S \subseteq N$ to $T\subseteq C$.

Consider a price system $\textsf{ps}$ that witnesses that $T$ satisfies $\mathcal{P}_{\mathrm{price\text{-}eq}}$ for the profile $P|_S$. For each candidate $c \in T$ by $\phi_c$ we denote the payment that each voter pays for $c$ in $\textsf{ps}$.
Let $W_t$ be the set of candidates selected by Equal Shares up to the $t$-th iteration.
We will first prove the following invariant: For each $t$, in the $t$-th iteration Equal Shares selects a candidate for which each voter pays at most $\min_{c \in T \setminus W_t} \phi_c$. For the sake of contradiction assume this is not the case and let $t$ be the first iteration in which Equal Shares selects a candidate $c_t$ for which some voter pays more than $\min_{c \in T \setminus W_t} \phi_c$. Let $c_t' = \argmin_{c \in T \setminus W_t} \phi_c$. We will argue that Equal Shares would rather select $c_t'$ than $c_t$.

Indeed, up to the $t$-th iteration each voter has fewer representatives in $W_t$ than in $T$ (since $T$ is a core deviation). Further, up to the $t$-th iteration each voter pays on average less for her representatives in $W_t$ than in $T$ (this is because $t$ is the first iteration where the candidate selected by Equal Shares costs the voter more than each not-yet selected candidate from $T$). Thus, each voter who pays for $c_t'$ in $\textsf{ps}$ has still some money left to pay for $c_t'$ (since they had money to pay for $c_t'$ in $\textsf{ps}$). Since Equal Shares always selects a candidate who is affordable and who minimizes the per-voter payment, it would rather select $c_t'$ than $c_t$, a contradiction; this proves our invariant.

However, by a very similar reasoning we can reach a contradiction. Indeed, by our invariant each voter pays for her representatives in $W$ less than it would pay for her representatives in $T$ according to $\textsf{ps}$. Each voter form $S$ has fewer representatives, and thus they would have money to buy some not-yet selected candidate from $T$. This gives a contradiction and completes the proof.
\end{proof} 

In \Cref{prop:x-not-core-priceable} in the appendix, we show that Equal Shares does not satisfy the core when we further relax the constraint to allow any deviations that are priceable, without requiring equal payments. In particular, Equal Shares fails the core (with unconstrained deviations).

\subsection{No welfarist rule satisfies the core}

Just like for laminar proportionality and priceability, we can prove that no welfarist rule can satisfy the core (in its full strength). The construction is significantly more involved than the previous arguments of this type. We see this result as a kind of obstruction to proving that the core is non-empty for all profiles, since we rule out the large class of welfarist rules as possible candidates.

\begin{figure}
\begin{center}
Profile 1:
\vspace{0.3cm}

\begin{tikzpicture}

\foreach \i in {1,...,16}
	\node at (-0.325 + 0.65 * \i, -0.35) {$v_{\i}$};

\filldraw[fill=white, draw=black] (0,0) rectangle (2.6,0.5);
\filldraw[fill=white, draw=black] (2.6,0) rectangle (3.9,0.5);
\filldraw[fill=white, draw=black] (3.9,0) rectangle (5.2,0.5);
\filldraw[fill=white, draw=black] (5.2,0) rectangle (6.5,0.5);
\filldraw[fill=white, draw=black] (6.5,0) rectangle (7.8,0.5);
\filldraw[fill=white, draw=black] (7.8,0) rectangle (9.1,0.5);
\filldraw[fill=white, draw=black] (9.1,0) rectangle (10.4,0.5);

\foreach \y in {0.5, 1, 1.5, 2.0} {
\filldraw[fill=white, draw=black] (0,\y) rectangle (2.6,\y+0.5);
\filldraw[fill=white, draw=black] (2.6,\y) rectangle (3.9,\y+0.5);
\filldraw[fill=white, draw=black] (3.9,\y) rectangle (5.2,\y+0.5);
\filldraw[fill=white, draw=black] (5.2,\y) rectangle (6.5,\y+0.5);
\filldraw[fill=white, draw=black] (6.5,\y) rectangle (7.8,\y+0.5);
\filldraw[fill=white, draw=black] (7.8,\y) rectangle (9.1,\y+0.5);
\filldraw[fill=white, draw=black] (9.1,\y) rectangle (10.4,\y+0.5);
}

\begin{scope}[yshift=1cm]
\filldraw[fill=white, draw=black] (0,1.5) rectangle (2.6,2.0);
\filldraw[fill=white, draw=black] (2.6,1.5) rectangle (3.25,2.0);
\filldraw[fill=white, draw=black] (3.25,1.5) rectangle (3.9,2.0);
\filldraw[fill=white, draw=black] (3.9,1.5) rectangle (4.55,2.0);
\filldraw[fill=white, draw=black] (4.55,1.5) rectangle (5.2,2.0);
\filldraw[fill=white, draw=black] (5.2,1.5) rectangle (5.85,2.0);
\filldraw[fill=white, draw=black] (5.85,1.5) rectangle (6.5,2.0);
\filldraw[fill=white, draw=black] (6.5,1.5) rectangle (7.15,2.0);
\filldraw[fill=white, draw=black] (7.15,1.5) rectangle (7.8,2.0);
\filldraw[fill=white, draw=black] (7.8,1.5) rectangle (8.45,2.0);
\filldraw[fill=white, draw=black] (8.45,1.5) rectangle (9.1,2.0);
\filldraw[fill=white, draw=black] (9.1,1.5) rectangle (9.75,2.0);
\filldraw[fill=white, draw=black] (9.75,1.5) rectangle (10.4,2.0);

\filldraw[fill=white, draw=black] (0,2.0) rectangle (0.65,2.5);
\filldraw[fill=white, draw=black] (0.65,2.0) rectangle (1.3,2.5);
\filldraw[fill=white, draw=black] (1.3,2.0) rectangle (1.95,2.5);
\filldraw[fill=white, draw=black] (1.95,2.0) rectangle (2.6,2.5);

\node at (0.325, 2.25) {$c_1$};
\node at (0.975, 2.25) {$c_2$};
\node at (1.625, 2.25) {$c_3$};
\node at (2.275, 2.25) {$c_4$};
\end{scope}

\node at (3.25, 0.25) {$c_{11}$};
\foreach \i in {5,6,7,8,9,10}
	\node at (1.3, 5.25 - 0.5*\i) {$c_{\i}$};
\end{tikzpicture}

Profile 2:
\vspace{0.3cm}

\begin{tikzpicture}
\foreach \i in {1,...,16}
	\node at (-0.325 + 0.65 * \i, -0.35) {$v_{\i}$};

\foreach \y in {0,0.5,1.0,1.5} {
\filldraw[fill=blue!10!white, draw=black] (0,\y) rectangle (1.3,\y+0.5);
\filldraw[fill=blue!10!white, draw=black] (1.3,\y) rectangle (2.6,\y+0.5);
\filldraw[fill=blue!10!white, draw=black] (2.6,\y) rectangle (3.9,\y+0.5);
\filldraw[fill=blue!10!white, draw=black] (3.9,\y) rectangle (5.2,\y+0.5);
\filldraw[fill=blue!10!white, draw=black] (5.2,\y) rectangle (6.5,\y+0.5);
\filldraw[fill=blue!10!white, draw=black] (6.5,\y) rectangle (7.8,\y+0.5);
\filldraw[fill=blue!10!white, draw=black] (7.8,\y) rectangle (9.1,\y+0.5);
\filldraw[fill=blue!10!white, draw=black] (9.1,\y) rectangle (10.4,\y+0.5);
}

\begin{scope}[yshift=1cm]
\filldraw[fill=blue!10!white, draw=black] (0,1.0) rectangle (1.3,1.5);
\filldraw[fill=blue!10!white, draw=black] (1.3,1.0) rectangle (2.6,1.5);
\filldraw[fill=blue!10!white, draw=black] (2.6,1.0) rectangle (3.9,1.5);
\filldraw[fill=blue!10!white, draw=black] (0,1.5) rectangle (1.3,2.0);
\filldraw[fill=white, draw=black] (1.3,1.5) rectangle (2.6,2.0);
\filldraw[fill=blue!10!white, draw=black] (2.6,1.5) rectangle (3.9,2.0);
\filldraw[fill=blue!10!white, draw=black] (0,2.0) rectangle (1.3,2.5);

\filldraw[fill=blue!10!white, draw=black] (3.9,1.0) rectangle (5.2,1.5);
\filldraw[fill=blue!10!white, draw=black] (5.2,1.0) rectangle (6.5,1.5);
\filldraw[fill=blue!10!white, draw=black] (6.5,1.0) rectangle (7.8,1.5);
\filldraw[fill=blue!10!white, draw=black] (7.8,1.0) rectangle (9.1,1.5);
\filldraw[fill=blue!10!white, draw=black] (9.1,1.0) rectangle (10.4,1.5);

\filldraw[fill=white, draw=black] (3.9,1.5) rectangle (5.2,2.0);
\filldraw[fill=blue!10!white, draw=black] (5.2,1.5) rectangle (6.5,2.0);
\filldraw[fill=blue!10!white, draw=black] (6.5,1.5) rectangle (7.8,2.0);
\filldraw[fill=white, draw=black] (7.8,1.5) rectangle (9.1,2.0);
\filldraw[fill=white, draw=black] (9.1,1.5) rectangle (10.4,2.0);

\filldraw[fill=blue!10!white, draw=black] (9.1,2.0) rectangle (9.75,2.5);

\filldraw[fill=blue!10!white, draw=black] (3.25,2.0) rectangle (4.55,2.5);
\end{scope}

\foreach \i in {1,...,7}
	\node at (0.65, 3.75-0.5*\i) {$c_\i$};
\foreach \i in {8,...,13}
	\node at (1.95, 3.25-0.5*\i+3.5) {$c_{\i}$};
\end{tikzpicture}

Profile 3:
\vspace{0.3cm}

\begin{tikzpicture}
\filldraw[fill=blue!10!white, draw=black] (0,0) rectangle (1.3,0.5);
\filldraw[fill=blue!10!white, draw=black] (1.3,0) rectangle (2.6,0.5);
\filldraw[fill=blue!10!white, draw=black] (2.6,0) rectangle (3.9,0.5);
\filldraw[fill=blue!10!white, draw=black] (0,1.0) rectangle (1.3,1.5);
\filldraw[fill=blue!10!white, draw=black] (1.3,1.0) rectangle (2.6,1.5);
\filldraw[fill=blue!10!white, draw=black] (2.6,1.0) rectangle (3.9,1.5);
\filldraw[fill=blue!10!white, draw=black] (0,1.5) rectangle (1.3,2.0);
\filldraw[fill=white, draw=black] (1.3,1.5) rectangle (2.6,2.0);
\filldraw[fill=blue!10!white, draw=black] (2.6,1.5) rectangle (3.9,2.0);
\filldraw[fill=blue!10!white, draw=black] (0,2.0) rectangle (1.3,2.5);

\node at (0.325, -0.35) {$v_1$};
\node at (0.975, -0.35) {$v_2$};
\node at (1.625, -0.35) {$v_3$};
\node at (2.275, -0.35) {$v_4$};
\node at (2.925, -0.35) {$v_5$};
\node at (3.575, -0.35) {$v_6$};
\node at (4.225, -0.35) {$v_7$};
\node at (4.875, -0.35) {$v_8$};

\node at (10.075, -0.35) {$v_{16}$};
\node at (7.475, -0.35) {$\ldots$};

\filldraw[fill=blue!10!white, draw=black] (3.9,0) rectangle (5.2,0.5);
\filldraw[fill=blue!10!white, draw=black] (5.2,0) rectangle (6.5,0.5);
\filldraw[fill=blue!10!white, draw=black] (6.5,0) rectangle (7.8,0.5);
\filldraw[fill=blue!10!white, draw=black] (7.8,0) rectangle (9.1,0.5);
\filldraw[fill=blue!10!white, draw=black] (9.1,0) rectangle (10.4,0.5);

\node at (0.65, 0.75) {$\cdots$};
\node at (1.95, 0.75) {$\cdots$};
\node at (3.25, 0.75) {$\cdots$};
\node at (4.55, 0.75) {$\cdots$};
\node at (5.85, 0.75) {$\cdots$};
\node at (7.15, 0.75) {$\cdots$};
\node at (8.45, 0.75) {$\cdots$};
\node at (9.75, 0.75) {$\cdots$};
\filldraw[fill=blue!10!white, draw=black] (3.9,1.0) rectangle (5.2,1.5);
\filldraw[fill=blue!10!white, draw=black] (5.2,1.0) rectangle (6.5,1.5);
\filldraw[fill=blue!10!white, draw=black] (6.5,1.0) rectangle (7.8,1.5);
\filldraw[fill=blue!10!white, draw=black] (7.8,1.0) rectangle (9.1,1.5);
\filldraw[fill=blue!10!white, draw=black] (9.1,1.0) rectangle (10.4,1.5);

\filldraw[fill=blue!10!white, draw=black] (3.9,1.5) rectangle (5.2,2.0);
\filldraw[fill=blue!10!white, draw=black] (5.2,1.5) rectangle (6.5,2.0);
\filldraw[fill=blue!10!white, draw=black] (6.5,1.5) rectangle (7.8,2.0);
\filldraw[fill=white, draw=black] (7.8,1.5) rectangle (9.1,2.0);
\filldraw[fill=white, draw=black] (9.1,1.5) rectangle (10.4,2.0);

\filldraw[fill=blue!10!white, draw=black] (9.1,2.0) rectangle (9.75,2.5);

\node at (0.65, 2.25) {$c_1$};
\node at (0.65, 1.75) {$c_2$};
\node at (0.65, 1.25) {$c_3$};
\node at (0.65, 0.25) {$c_7$};
\node at (1.95, 1.75) {$c_8$};
\node at (1.95, 1.25) {$c_9$};
\node at (1.95, 0.25) {$c_{13}$};
\end{tikzpicture}

\end{center}
\caption{A diagram illustrating the reasoning in \Cref{thm:no_welfarist_core}, showing that there exists no welfarist rule that satisfies the core property. Here each rectangle denotes a single candidate who is approved by the respective voters. For example, in Profile 1 candidates $c_5$ and $c_6$ are approved by $\{v_1, v_2, v_3, v_4\}$ and $c_{11}$ is approved by $\{v_5, v_6\}$. Profile 2 depicts the case for $x = 3$ and $y = 6$. Profile 3 depicts the case for $x = 2$ and $y = 5$.}
\label{fig:core_and_welfarism}
\end{figure}

\begin{theorem}\label{thm:no_welfarist_core}
There exists no welfarist committee rule that satisfies core.
\end{theorem}
\begin{proof}
Assume that there exists a welfarist committee rule $f$ that satisfies core.

Consider Profile 1 from \Cref{fig:core_and_welfarism} with 16 voters and 52 candidates, and let $k = 48$. To this profile we add two candidates who are not approved by any voter, getting 54 candidates in total.
The group of three voters $\{v_1, v_2, v_3\}$ deserves to decide about $48 \cdot \nicefrac{3}{16} = 9$ committee members. Since this group approves candidates $\{c_1, c_2, c_3, c_5, \ldots, c_{10}\}$, by the core property we get that $c_5, c_6, \ldots, c_{10}$, and at least one of the candidates $c_1, c_2, c_3$ must be members of the winning committee. Without loss of generality, assume that $c_1$ is picked by $f$. By the same reasoning, applied to $\{v_2, v_3, v_4\}$ we infer that at least one of the candidates $c_2, c_3, c_4$ must be included in the winning committee. Assume $c_2$ is picked. Thus, in the winning committee: some $x$ voters ($x \in \{2, 3, 4\}$) have 7 representatives each. In the remaining part of the proof we will assume that $x=3$, and that the voters who have 7 representatives are $v_1, v_2$ and $v_3$. The reasoning when $x$ is even is simpler, and we will shortly discuss it at the end of the proof.

By looking at $\{v_5, v_6\}$ (and again applying the core property) we get that either:
\begin{inparaenum}[(i)]
\item both voters have 6 representatives (this happens when the rule picks all 7 candidates approved by $v_5$ and $v_6$),
\item one of the voters has 6 representatives and the other one has 5 (this is obtained by picking 6 candidates out of the seven approved by $v_5$ and $v_6$),
\item both voters have 5 representatives (picking 5 or 6 candidates out of the seven approved by $v_5$ and $v_6$).
\end{inparaenum}
Summing up, we get that the winning committee must consist of candidates $c_5, c_6, \ldots, c_{10}$, candidates $c_1, c_2, c_3$, and for each pair of voters $\{v_{2i-1}, v_{2i}\}$ with $3 \leq i \leq 8$, the committee must have at least 5 candidates approved by these two voters.  This gives in total 39 candidates, and leaves $9$ free seats in the committee.
Thus, at least $12 - 9 = 3$ candidates who are approved by voters $v_5, v_8, \ldots, v_{16}$ are not in the winning committee. Finally, only $54 - 48 = 6$ candidates do not get in to the committee. Since we assumed $x=3$, one of the candidates that do not get in to the committee is $c_4$. As a result, at least two voters from $v_5, v_6, \ldots, v_{16}$ will have 6 representatives.

Consequently, in the winning committee: $3$ voters have 7 representatives each, some $y$ voters ($3 \leq y \leq 10$) have 5 representatives each; the remaining voters (at least 2 from $v_5, v_6, \ldots, v_{16}$ and $v_4$) have 6 representatives. Let $\mathbf{w_1}$ be a welfare vector corresponding to the winning committee.  Further, without loss of generality assume that in the winning committee $v_{15}$ and $v_{16}$ have 6 representatives, each.

Now, consider Profile 2 from \Cref{fig:core_and_welfarism}. In this profile the voters are grouped in pairs. Voters $v_1$ and $v_2$ approve some 7 candidates; voters $v_3$ and $v_4$ approve some 6 candidates. The pairs $(v_5, v_6)$ and $(v_7, v_8)$ approve some 6 candidates (different for each pair); we add one additional candidate, who is approved by $v_6$ and $v_7$.  The next four pairs approve some 6 candidates (different for each pair). Additionally, the first voter in the last pair (voter $v_{15}$) approves one more different candidate. Finally, there are three candidates who are not approved by any voter---we call them dummy candidates. 

Consider the committee (for Profile 2) assembled in the following way: we take the 7 candidates approved by the first pair, 5 candidates approved by the second pair, 6 candidates approved by $(v_5, v_6)$, 5 candidates approved by $(v_7, v_8)$, and the candidate approved by $(v_6, v_7)$. This already gives us: $3$ voters with 7 representatives, $3$ voters with 5 representatives and 2 voters with 6 representatives. 
Further, for each of the next three pairs we take either:
\begin{inparaenum}[(i)]
\item 5 candidates approved by both the voters in the pair and one dummy candidate (getting two voters with 5 representatives, each), or
\item 6 candidates approved by both the voters in the pair (getting two voters with 6 representatives, each).
\end{inparaenum}
Finally, for the last pair, we take either:
\begin{inparaenum}[(i)]
\item 5 candidates approved by both the voters in the pair and the one approved by $v_{15}$ (getting one voter with 5 representatives and one with 6), or
\item 6 candidates approved by both the voters in the pair (getting two voters with 6 representatives, each).
\end{inparaenum}
By this procedure, it is possible to assemble a committee with $3$ sevens, $y$ fives, and sixes in the remaining (at least three) positions (one can see that the last pair is constructed differently so that we can assemble such a committee for every parity of $y$). This committee (for $x = 3$, $y = 6$) is marked with a color in the picture.

By taking the appropriate permutation $\sigma$ of the voters in Profile 2 we can make this welfare vector match $\mathbf{w_1}$.  This shows that the welfare vector $\mathbf{w_1}$ is achievable in (permuted) Profile 2.
Specifically, we will use the permutation which satisfies the following properties. It matches voters $v_5, v_8$, and $v_{16}$ from Profile 2 to some voters from $v_5, v_6, \ldots, v_{16}$ in Profile 1---this is possible since in Profile 1 there is a sufficient number of voters who have 6 representatives among those from $v_5, v_6, \ldots, v_{16}$ (to match $v_5$ and $v_{16}$, if $v_{16}$ has 6 representatives), and there is a sufficient number of voters who have 5 representatives among those from $v_5, v_6, \ldots, v_{16}$ (to match $v_8$ and $v_{16}$, if $v_{16}$ has 5 representatives). Further either $v_5$ and $v_{16}$ or $v_8$ and $v_{16}$ should be matched to two voters from the same pair in Profile 1 (this is possible for the following reason: If $v_{16}$ has 5 representatives, then these two pairs, $(v_5, v_{16})$, and $(v_8, v_{16})$ have, respectively $(6, 5)$ and $(5, 5)$ representatives. At least one pair in Profile 1 has $(5, 6)$, $(6, 5)$, or $(5, 5)$ representatives---this is because at least one voter among the last 12 in Profile 1 has 5 representatives and at least two have 6----so the matching is possible. If $v_{16}$ has 6 representatives, the reasoning is similar: the two pairs, $(v_5, v_{16})$, and $(v_8, v_{16})$ have, respectively $(6, 6)$ and $(5, 6)$ representatives, and at least one pair in Profile 1 has $(5, 6)$, $(6, 5)$, or $(6, 6)$ representatives). 
Further, the permutation matches $v_6$ with $v_3$, and $v_{15}$ with $v_4$ (note that $v_{15}$ in our assembled committee has always 6 representatives). Finally, the permutation matches the first 2 voters in Profile 1 with the same voters in Profile 2.

Yet, from the core property it follows that the winning committee in Profile 2 must give a different welfare vector. Such a committee needs to have one of the following forms. The first four voters get at least 6 representatives. The numbers of the representatives that the next four voters get (voters $v_5, v_6, v_7$ and $v_8$) can be either, $(5, 6, 6, 5)$, $(6, 6, 6, 6)$, $(6, 7, 6, 5)$, or $(5, 6, 7, 6)$. Each voter from the next three pairs gets 6 representatives. The two voters in the last pair get either 6 representatives each, or $v_{15}$ gets 6 representatives and $v_{16}$ gets 5. Finally, if $(v_5, v_6, v_7, v_8)$ gets $(5, 6, 6, 5)$ representatives, then there is one free slot in the committee which can be used for a dummy candidate, for a candidate approved by the first pair, or by someone approved by a voter in the last pair.

One can check that (with permutation $\sigma$) this welfare vector, call it $\mathbf{w_2}$ is achievable in Profile 1. For this, we need to observe that this welfare vector has at most 3 fives (and if there are exactly 3 fives, then two of them correspond to the voters matched to a single pair in Profile 1). In Profile 1 we can first take candidates $c_5, \ldots, c_{10}$, and for each pair of voters $\{v_{2i-1}, v_{2i}\}$ with $3 \leq i \leq 8$ we can take the seven candidates approved by these voters. This gives a vector with all sixes. Now, we can remove at most two candidates to match the ``fives'' in the welfare vector (this is possible by removing only two candidates since two voters among those corresponding to fives are matched to a single pair in Profile 1), and add dummy candidates or candidates from $c_1, \ldots, c_4$ to match the ``sevens'' (this is also possible since the voters who can have seven representatives---the voters in the first pair, $v_6$, and $v_{15}$ are matched to the voters from $v_1, \ldots, v_4$ in Profile 1). 

This vector is different from $\mathbf{w_1}$ since it has 7 in at most two positions.

In Profile 1 $f$ prefers $\mathbf{w_1}$ over $\mathbf{w_2}$, while in the permuted Profile 2 the other way around. This gives a contradiction and completes the proof for the case of $x = 3$.

If $x$ is even, then instead of Profile 2 we use Profile 3 from \Cref{fig:core_and_welfarism}. The proof follows from a similar reasoning, which now is much simpler since we do not need to deal with the candidate commonly accepted by $v_6$ and $v_7$.
\end{proof}

\section{Conclusion}

Through a detailed axiomatic study, we have obtained a clearer picture of the way that Phragm\'en's voting rule provides proportional representation. We have seen that these proportionality properties cannot be achieved by welfarist rules such as PAV. Conversely, proportionality in the sense of Phragm\'en's rule is incompatible with the basic fairness condition of the Pigou--Dalton principle. We showed that PAV gives an optimal approximation to the core, subject to that principle. Finally, we have introduced an elegant new voting rule that combines the proportionality properties of Phragm\'en's rule while also satisfying EJR. 

Our study distinguishes two conceptually different types of fairness/proportionality. We see axioms that describe fairness with respect to the distribution of welfare (take as examples the Pigou--Dalton principle, proportionality degree~\citep{skowron:prop-degree}, or extended justified representation) or with respect to the distribution of power (e.g., priceability, and laminar proportionality). This dichotomy of axioms is reflected in a dichotomy of voting rules: PAV is a welfarist rule that ensures a fair distribution of voter satisfaction, whereas Phragm\'en's rule and Equal Shares primarily aim at achieving a fair distribution of voting power. One can observe a similar dichotomy in other areas of voting theory: Condorcet's criterion implies an equal distribution of voting power, while Borda's rule is welfarist and may overrule a majority.\footnote{We thank Herv\'e Moulin for this observation and discussion on this resemblance.} When viewed from this perspective, the dispute between Thiele and Phragm\'en looks more similar to the one between Borda and Condorcet than it might at first appear.

Parts of our discussion has been relegated to the appendix. In particular, in \Cref{sec:overlapping_parties} we highlight a few differences between our new method and Phragm\'en's voting rule, and in \Cref{sec:proportionality_and_disagreements} we provide an additional informal discussion on the two types of proportionality offered by PAV and priceable rules.

Several important and technically interesting questions remain. Known priceable rules (such as Phragm\'en's rule) fail Pareto-optimality, and the natural way to achieve Pareto-optimality (some form of welfare maximization) cannot guarantee priceability. Does there always exist a Pareto-optimal and priceable committee? Is there a natural voting rule satisfying both properties? We believe that a detailed analysis of the class of priceable rules is an interesting direction for further research. Further, we have seen that no welfarist rule satisfies the core. Can one show that no welfarist rule approximates the core to a better factor than PAV? We have shown that Equal Shares satisfies the core if we only allow deviations that treat members of the deviating coalition fairly. Is it possible to prove the existence of the core with weaker constraints than the ones we identified (e.g., priceability without the requirement of equal payments)?

\subsection{Acknowledgments}
We thank Herv\'e Moulin for useful comments. Piotr Skowron was supported by the Foundation for Polish Science within the Homing programme (Project title: "Normative Comparison of Multiwinner Election Rules).

\bibliographystyle{plainnat}

\appendix

\section{Additional Discussion}\label{sec:additional_discussion}

This section provides some additional explanations that have been omitted from the main text. 

The following remark justifies certain decisions that we made when formalizing the concept of laminar proportionality---in particular it explains our focus  on \emph{integral} profiles.

	\begin{remark}\label{rem:integrality_in_laminarity_def}
		One could strengthen a part of the definition of laminar proportionality as follows: If there are two instances $(P_1, k_1)$ and $(P_2, k_2)$, \emph{not necessarily laminar}, with $C(P_1) \cap C(P_2) = \emptyset$ and $|P_1|/k_1 = |P_2|/k_2$, then for the instance $(P_1 + P_2, k_1 + k_2)$, we must have $f(P, k) = \{ W_1 \cup W_2 : W_1 \in f(P_1, k_1), W_2 \in f(P_2, k_2) \}$. This condition is failed by many laminar proportional rules. Consider the instance $(P_1, k=2)$ where 2 voters approve $\{c_1\}$, 2 voters approve $\{c_2\}$ and 2 voters approve $\{c_3\}$; and the instance $(P_2, k=4)$ where 12 voters approve $\{c_4, c_5, c_6, c_7, c_8\}$. In both instances, Phragm\'en's rule (or any reasonable rule) declares a tie between all feasible committees. The combined instance $(P_1 + P_2, k=6)$ is a party-list profile, and Phragm\'en's rule selects a committee of form $\{c_i, c_4, c_5, c_6, c_7, c_8\}$ where $c_i \in \{c_1,c_2,c_3\}$, failing the strengthened condition. The same is true for any other rule which behaves like D'Hondt's apportionment method on party-list profiles. From Theorem~\ref{thm:price-dhondt} it follows that this strengthened condition is incompatible with priceability.
	\end{remark}

\Cref{prop:laminarity_of_set_system} below, explains the name ``laminar proportionality''.

	\begin{proposition}\label{prop:laminarity_of_set_system}
		Suppose $(P,k)$ is laminar. For each $c \in C(P)$, let $N_{c} = \{ i \in N : c \in A_i\}$ be the set of voters approving $c$. Then the family $(N_c)_{c \in C(P)}$ is laminar.
	\end{proposition}
	\begin{proof}
		We prove this by induction on $k$.
		(i) If $P$ is unanimous, then the sets $N_c$ are all equal.
		(ii) Suppose $c' \in C(P)$ is such that $c' \in A_i$ for all $A_i \in P$, and $(P_{-c'}, k-1)$ is laminar. Then by the inductive hypothesis, the set system $(N_c)_{c \in C(P) \setminus \{c'\}}$ is laminar. Further, $N_{c'} = N(P)$, and hence $N_c \subseteq N_{c'}$ for all $c \in C(P) \setminus \{c'\}$. Thus, the set system $(N_c)_{c\in C(P)}$ is laminar, as required.
		(iii) Suppose $P = P_1 + P_2$ for some laminar $(P_1, k_1)$ and $(P_2, k_2)$ on disjoint candidate sets. Let $c_1 \in C(P_1)$ and $c_2 \in C(P_2)$. Then $N_{c_1}$ and $N_{c_2}$ are disjoint. Since $(N_c)_{c \in C(P_1)}$ and $(N_c)_{c \in C(P_2)}$ are laminar by the inductive hypothesis, the overall set system $(N_c)_{c \in C(P)}$ is also laminar.
	\end{proof}
	
\subsection{Overlapping Parties: Comparing Equal Shares and Phragm\'en's Sequential Rule}\label{sec:overlapping_parties}
	
This section highlights some differences between Equal Shares and Phragm\'en's Sequential Rule. As we have already noticed, Equal Shares satisfies EJR and core subject to priceability with equal payments. On the other hand, Phragm\'en's rule is committee-monotonic.\footnote{If we increase the size of the committee from $k$ to $k+1$, then the new committee selected by Phragm\'en's rule will have the same members as the old committee of size $k$, plus one additional candidate.} Below, we give an example of a simple class of profiles (overlapping parties) where the behavior of the two rules differs significantly. 

Consider the following profile:

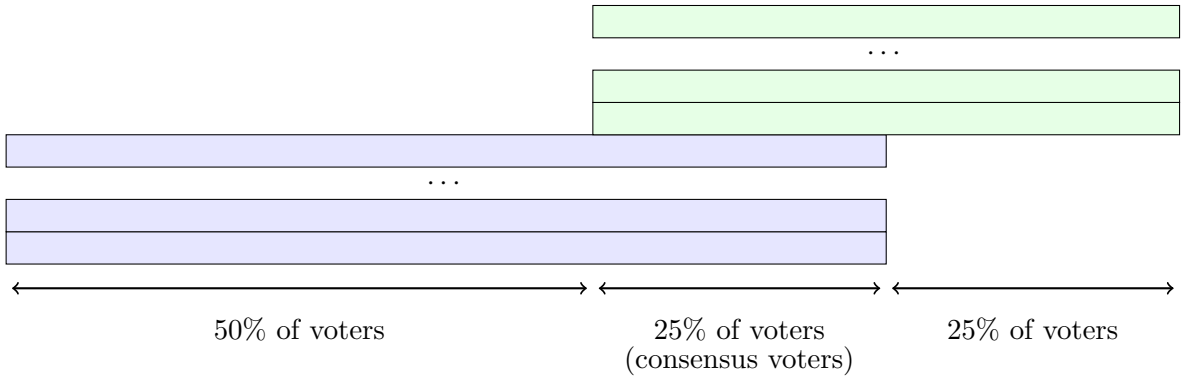
\begin{center}
\begin{tikzpicture}
[yscale=0.43,xscale=0.78,voter/.style={anchor=south, yshift=-7pt}, party1/.style={fill=blue!10}, party2/.style={fill=green!10}, c/.style={anchor=south, yshift=1.5pt, inner sep=0}]
	\draw[party1] (0,0) rectangle (15,1);
	\draw[party1] (0,1) rectangle (15,2);
	\node at (7.5, 2.5) {$\cdots$};
	\draw[party1] (0,3) rectangle (15,4);
	\draw[party2] (10.0, 4) rectangle (20,5);
	\draw[party2] (10.0, 5) rectangle (20,6);
	\node at (15, 6.5) {$\cdots$};
	\draw[party2] (10.0, 7) rectangle (20,8);
	
	\draw[thick,<->] (0.1, -0.75) -- (9.9, -0.75);
	\node at (5.0, -2) {50\% of voters};
	
	\draw[thick,<->] (10.1, -0.75) -- (14.9, -0.75);
	\node at (12.5, -2) {25\% of voters};
	\node at (12.5, -3) {(consensus voters)};
			
	\draw[thick,<->] (15.1, -0.75) -- (19.9, -0.75);
	\node at (17.5, -2) {25\% of voters};
\end{tikzpicture}
\end{center}

We have two overlapping parties, marked with two different colors. The first 50\% of voters approve only the first party, the next 25\% of voters approve both parties, and the last 25\% of voters approve only the second party. Assume $k = 100$, and that each party has at least 100 candidates. In this example the Method of Equal Shares would pick 75\% of candidates from the first party, and 25\% from the second one.\footnote{Interestingly, this example can also be interpreted as ``laminar''. Here the approval sets of the voters form a laminar family. This is different from the case of laminar profiles, where laminarity is with respect to the family of sets $\{N_{c}\}_{c \in C}$ (see \Cref{prop:laminarity_of_set_system}).} Phragm\'en's rule, on the other hand, would assign roughly $\nicefrac{2}{3}$ seats to the first party and $\nicefrac{1}{3}$ seats to the second party (this would be also the outcome of PAV run for this profile). 

Intuitively, Phragm\'en's rule (and PAV) notices that the consensus voters are satisfied by any outcome, and implements proportionality with respect to the remaining voters who approve only one of the two parties. The proportion of these two groups of voters is $50/25$ and so this is also the proportion of the seats assigned to the respective parties by Phragm\'en's rule (and PAV). Equal Shares treats the consensus voters in a different way. It observes that they can be treated either as supporters of the first party, or as the supporters of the second one. Equal Shares makes the choice how to treat these voters based on the utilitarian criterion: indeed, in the above example the consensus voters are treated as supporters of the larger party, and with this assumption the solution returned by Equal Shares is also proportional (and has a higher total support from the voters).
	
\subsection{Proportionality with Respect to Disagreements}\label{sec:proportionality_and_disagreements}

In this subsection we give some informal insights on the two types of proportionality, implemented by priceable rules (here, for the sake of concreteness, we use the example of Equal Shares) and PAV.

Very informally speaking, in Equal Shares each voter has a voting power that is equal to the amount of money she holds. For example, consider two groups of voters, $S_1$ and $S_2$, and assume $|S_2| = \gamma |S_1|$: then, since group $S_2$ is $\gamma$ times larger than $S_1$, it initially has $\gamma$ times more money, and so it has $\gamma$ times more voting power. This intuitively explains why Equal Shares is proportional. Now, assume that there is a group $T$ of $\alpha$ committee members that the voters from $S_1$ and $S_2$ agree on (either all members of $S_1 \cup S_2$ approve them, or all members disapprove them). Since including a candidate in a committee costs each voter the same amount of money (apart form the corner cases, when the voter pays all her left money), the voters from $S_2$ must have paid for the candidates from~$T$ roughly $\gamma$ times more money than the voters from~$S_1$. Thus, the money that the groups~$S_1$ and~$S_2$ have excluding the money that they paid for candidates from $T$ is still in proportion $1:\lambda$. Thus, the relative voting power of the two groups did not change as a result of including in the committee candidates that the two groups agree on. Intuitively Equal Shares implements a kind of proportionality that focuses on how to fairly resolve \myemph{the conflicts of opinions} in a group of voters, and thus when dividing the seats between two groups of voters (relatively, leaving the remaining voters aside) it does not consider the candidates that the two groups agree on. 

This is different from PAV. For PAV the concept of a voting power of a voter is much harder to grasp, thus we will rather focus on its sequential variant.\footnote{In sequential PAV we start with and empty committee $W = \emptyset$ and in each of $k$ consecutive steps we add one candidate to $W$. The candidate to be added is the one that increases the PAV score of the committee $W$ most.} In this case when a voter already has $r$ representatives in a committee, her voting power (again, intuitively, and very informally) equals $\nicefrac{1}{r_i + 1}$. Initially, the voting power of the groups $S_1$ and $S_2$ from the previous example is in proportions $1 : \lambda$. That is, looking only at these groups alone, if a group $S_1$ deserves $x$ representatives, none of which is liked by $S_2$, then group $S_2$ deserves $\lambda x$ representatives---this is because $1/x$ = $\lambda/(\lambda x)$. However, including in the committee $\alpha$ candidates that the voters from $S_1 \cup S_2$ approve changes this ratio of the voting power dramatically. Now, if group $S_1$ deserves $x$ candidates who are not liked by $S_2$, then $S_2$ deserves roughly $y$ representatives such that $\frac{1}{x + \alpha} = \frac{\lambda}{y + \alpha}$. Thus, the relative voting power of groups $S_1$ and $S_2$ changes as a result of including some commonly agreed candidates in the committee. This is because PAV (and its sequential counterpart) implements a type of proportionality that is concerned with the absolute welfare of the voters rather than the one implemented by Equal Shares which rather looks at the disagreements among the voters.

\section{Proofs Omitted From the Main Text}

\begin{samepage}
\begin{reptheorem}{thm:phragmen_satisfies_laminar_proportionality}
\phragmensatisfieslaminarproportionality
\end{reptheorem}
\end{samepage}
	\begin{proof}
		(a) Let us start by providing the proof for Phragm\'en's sequential rule. We prove, by induction on $k$, the stronger statement that on every laminar instance $(P,k)$, Phragm\'en's sequential rule is laminar proportional, and that it terminates at time $t = k/n$.
		
		(i) If $P$ is unanimous, then Phragm\'en's sequential rule elects any $k$ candidates in an arbitrary order, one candidate after each time increment of $1/n$, so it terminates at $t = k/n$.
		
		(ii) Let $C^* = \{c \in C(P) : c \in \bigcap_{i \in N} A_i\}$ be the set of candidates who are approved by every voter, and suppose that $r := |C^*| \ge 1$, and that $(P|_{C \setminus C^*}, k-r)$ is laminar. Note that in the profile $P|_{C \setminus C^*}$, there does not exist a candidate who is approved by every voter. Start running Phragm\'en's sequential rule until time $t = r/n$. At this point, all candidates in $C^*$ will have been elected (in some order), and all $n$ voters have spent all their money. Thus, from this point onward, Phragm\'en's sequential rule behaves exactly as if it was run on $P|_{C \setminus C^*}$. Hence, it outputs $C^* \cup W'$, where $W'$ is a committee output by Phragm\'en's sequential rule for $P|_{C \setminus C^*}$. By the inductive hypothesis, $W'$ is a laminar proportional output for $(P|_{C \setminus C^*}, k-r)$, and Phragm\'en's sequential rule terminates after time $(k-r)/n$. Thus,  $C^* \cup W'$ is a laminar proportional output for $(P,k)$, and is obtained after time $r/n + (k-r)/n = k/n$.
		
		(iii) Suppose $(P_1, k_1)$ and $(P_2, k_2)$ are disjoint laminar instances with $P = P_1 + P_2$, $k = k_1 + k_2$ and such that $n/k = n_1/k_1 = n_2/k_2$. Run Phragm\'en's sequential rule on $P$ until time $k/n$. Since the subprofiles $P_1$ and $P_2$ are disjoint (different voters and different candidates), the set of candidates $W$ elected until this time point is exactly the union of candidates elected by Phragm\'en's rule run on the two profiles separately. Let $W_1$ be the candidates elected for $P_1$, and $W_2$ those for $P_2$. By the inductive hypothesis, for $j=1,2$, we have $|W_j| = k_j$ because the rule terminates in time $k_j/n_j = k/n$, and $W_j$ is laminar proportional. Hence, $W = W_1 \cup W_2$ is laminar proportional, and Phragm\'en's sequential rule at time $k/n$.
		
		The proof establishes that on any laminar instance, Phragm\'en's sequential rule terminates in time $k/n$.

		(b) We prove that Equal Shares is laminar proportional by induction on $k$. Let $(P,k)$ be a laminar instance. (i) If $P$ is unanimous, then the rule is obviously laminar proportional on $(P,k)$. (ii) After electing a unanimously approved candidate, each voter's budget is reduced to $1-\frac1k$. According to the definition of Equal Shares, if a coalition $V$ of voters wishes to buy a candidate, this costs each voter $n/(k|V|)$. Now change currency by multiplying all remaining budgets and purchase costs by the factor $k/(k-1)$. Then each voter's budget is 1, and buying a candidate costs $n/(k-1)$. Thus, we can see that from now on Equal Shares proceeds as if it was run on the instance $(P_{-c}, k-1)$. (iii) Suppose $(P_1, k_1)$ and $(P_2, k_2)$ are disjoint laminar instances with $P = P_1 + P_2$, $k = k_1 + k_2$ and such that $n/k = n_1/k_1 = n_2/k_2$. Then it is clear that Equal Shares, run on $(P,k)$, proceeds exactly as if run in parallel on the instances $(P_1, k_1)$ and $(P_2, k_2)$, since by $n/k = n_1/k_1 = n_2/k_2$ the prices in all three instances are the same.
		
		(c) In \Cref{ex:laminar-two-parties}, the committee $\{c_1, c_2, c_3, c_4\}$ is one of the optimal committees under PAV, but it is not a laminar proportional committee. Refer to the introduction for another instance with a unique PAV committee that fails laminar proportionality.
	\end{proof}

\begin{reptheorem}{thm:rule_x_core}
\ruleXandCoreApproximation
\end{reptheorem}
	
	\begin{proof}
	\textbf{Upper bound.}
Let us fix $\alpha = \frac{1}{2\log(2k)}$. Let $A$ be an approval based profile and let $W$ be a committee returned by Equal Shares for $A$.
Recall that $p = \nicefrac{n}{k}$ is the price that Equal Shares uses to construct~$W$.
Further, assume that there exists a group of voters $S \subseteq N$ and a set of candidates $T$ such that:
\begin{enumerate}
\item $|T| \leq k \cdot \frac{|S|}{n}$, and 
\item $\frac{1}{4\log(2k) + 1} \cdot |A_i \cap T| > \max(|A_i \cap W|, 1)$ for each $i \in S$.
\end{enumerate}
Thus, the groups $S$ and $T$ witness that Equal Shares does not satisfy $\frac{1}{4\log(2k) + 1} $-core property. We will derive a contradiction from this assumption, which is sufficient to prove the theorem.
First, observe that since $\frac{1}{4\log(2k) + 1} = \frac{\alpha}{\alpha+2}$ the second condition implies that for each $i \in S$:
\begin{align}\label{eq:core_cond_rule_x}
\alpha \cdot \big(|A_i \cap T| - |A_i \cap W|\big) >  2|A_v \cap W| \text{.}
\end{align}
Further, if $|A_i \cap W| = 0$, then $\frac{1}{4\log(2k) + 1} \cdot |A_i \cap T| > 1$ implies that $\alpha \cdot \big(|A_i \cap T| - |A_i \cap W|\big) > 1$. In any case, we can reformulate \eqref{eq:core_cond_rule_x} as:
\begin{align}\label{eq:core_cond_rule_x2}
\alpha \cdot \big(|A_i \cap T| - |A_i \cap W|\big) > |A_i \cap W| + 1 \text{.}
\end{align}

Now, consider a fixed subset $S' \subseteq S$, and let 
\begin{align*}
s_{S'} = \sum_{i \in S'} \big(|A_i \cap T| - |A_i \cap W|\big).
\end{align*}

Consider all time moments during the execution of the Method of Equal Shares for profile $A$ which satisfy the following condition: each voter from $S'$ has at least $p \cdot \frac{|T|}{s_{S'}}$ dollars left (the choice of this value will become clear later on). From the pigeonhole principle, during the whole execution of Equal Shares---in particular, in the time moments that we consider---there exists a not-yet selected candidate who is approved by at least $\nicefrac{s_{S'}}{|T|}$ voters from $S'$. Thus, during the considered time moments the amount that a single voter needs to pay for a representative is no greater than:
\begin{align*}
 p \cdot \frac{|T|}{s_{S'}}
\end{align*}
Further, since in all the considered time moments the voters from $S'$ have enough money to buy an additional candidate, the procedure of assembling the committee cannot stop. Thus, there must exists a moment when some voter from $S'$, call it $i$, is left with less than $p \cdot \frac{|T|}{s_{S'}}$ dollars. Let $t$ be the first such a moment. 
As we observed, before $t$ each voter from $S'$ pays at most $p \cdot \frac{|T|}{s_{S'}}$ for each representative.
Thus, at time $t$ voter $i$ has more than the following number of representatives:
\begin{align*}
r_i = \frac{1 - p \cdot \frac{|T|}{s_{S'}}}{p \cdot \frac{|T|}{s_{S'}}} = \frac{s_{S'}}{p|T|} - 1 \text{.}
\end{align*}
By \eqref{eq:core_cond_rule_x2} we have that:
\begin{align*}
|A_i \cap T| - |A_i \cap W| \geq  \frac{|A_i \cap W| + 1}{\alpha} > \frac{r_i + 1}{\alpha} = \frac{s_{S'}}{p|T|\alpha}
\end{align*}
Let $S'' = S' \setminus \{i\}$. Clearly, we have that:
\begin{align*} 
s_{S''} = s_{S'} - (|A_i \cap T| - |A_i \cap W|) \leq s_{S'}\left(1 - \frac{1}{p|T|\alpha}\right) \text{.}
\end{align*}
The above reasoning holds for each $S' \subseteq S$. Thus, we start with $S' = S$ and apply it recursively, in each iteration decreasing the size of $S'$ by $1$. After $\nicefrac{|S|}{2}$ iterations we are left with a subset~$S_{e}$ such that:
\begin{align*} 
s_{S_e} \leq s_{S}\left(1 - \frac{1}{p|T|\alpha}\right)^{\frac{|S|}{2}} \leq s_{S}\left(1 - \frac{1}{p|T|\alpha}\right)^{\frac{n|T|}{2k}} = s_{S}\left(1 - \frac{1}{p|T|\alpha}\right)^{\frac{p|T|}{2}}  < s_{S}\left(\frac{1}{e}\right)^{\frac{1}{2\alpha}} \text{.}
\end{align*}  
Since $s_{S_e} \geq \nicefrac{|S|}{2}$ (for each $i \in S$ it must hold that $A_i \cap T \neq \emptyset$) and $s_{S} \leq |S|\cdot |T|$ we get that:
\begin{align*}
\frac{|S|}{2} \cdot e^{\frac{1}{2\alpha}} < |S|\cdot|T| \text{,}
\end{align*}
which is equivalent to $e^{\frac{1}{2\alpha}} < 2|T|$ and, further, to $\alpha \geq \frac{1}{2\log(2|T|)} > \frac{1}{2\log(2k)}$. This gives a contradiction and completes the proof for the upper bound.

\paragraph{Lower bound.}
Fix $c < 1$. We will construct a preference profile for which Equal Shares returns a committee that violates the core property up to a multiplicative ratio of the order of $O\left(\frac{1}{\log^c(k)}\right)$. In the construction we will use an intuitively large constant value $L$, which we do not specify yet, but which will become clear later on. Further, we do not specify the committee size $k$ nor the number of voters $n$ yet, but we fix their ratio to $\frac{n}{k} = L$. Thus, intuitively, a group of $L$ voters will have exactly enough money to pay for a single representative in the winning committee. Finally, the construction is parameterized with a constant $x$; the committee size will be selected so that $x = O(\log(k))$.  

First, we introduce $x$ groups of voters, $S_1, \ldots, S_x$, and the set $R$ of $x^x$ candidates. The groups $S_1, \ldots, S_x$ are equal-sized: $|S_i| = L\cdot x^{x-1}$ for each $i \in [x]$. Each voter from group $S_i$ approves exactly $x^i$ candidates from $R$; further, the approval sets are constructed so that within each group $S_i$ each candidate from $R$ is approved by the same number of voters. Such a construction is possible if $L$ is sufficiently large. Observe that in the so far construction, each candidate from $R$ is approved by the following number of voters:
\begin{align*}
s_{x} = \frac{\sum_{i = 1}^x |S_i| \cdot x^i}{x^x} = \frac{L}{x} \cdot \sum_{i = 1}^x \cdot x^i = \frac{L(x^{x} - 1)}{x - 1} \text{.} 
\end{align*}
In particular, observe that $s_{x} > |S_i|$. Now, we introduce $\frac{s_{x}}{L}$ additional candidates who are all approved by the same set of $s_{x}$ voters---call the set of these candidates $R_1$; the voters who approve $R_1$ are those from $S_x$ and $s_{x} - |S_i|$ new voters, not belonging to $S = S_1 \cup \ldots \cup S_x$. The idea is that Equal Shares run with price $p = \nicefrac{n}{k}$ first selects the candidates from $R_1$. Since $\frac{n}{k} = L$, after the candidates from $R_1$ are selected all the voters from $S_x$ will have no money left. 

We will construct the remaining part of the profile inductively. Assume that after some first steps of the execution of Equal Shares the voters from groups $S_{\ell+1}, S_{\ell+2}, \ldots, S_x$ are left without money, for some $\ell \in \{1, \ldots, x-1\}$. The number of approvals that each candidate from $R$ gets from the voters from $S_1 \cup \ldots \cup S_{\ell}$ is equal to:
\begin{align*}
s_{\ell} = \frac{\sum_{i = 1}^\ell |S_i| \cdot x^i}{x^x} = \frac{L}{x} \cdot \sum_{i = 1}^{\ell} x^i = \frac{L(x^{\ell} - 1)}{x - 1} \text{.} 
\end{align*}
Observe that for $\ell < x$ we have that $s_{\ell} < |S_i|$. Now we introduce additional candidates to our construction as follows. First, we take some $s_{\ell}$ voters from $S_{\ell}$ and introduce $\frac{s_{\ell}}{L}$ candidates who are approved by these voters. Equal Shares will select these candidates in the next step and leave these~$s_{\ell}$ voters without money. We remove these $s_{\ell}$ voters from~$S_{\ell}$ and we repeat the procedure until in $S_{\ell}$ there are less than $s_{\ell}$ voters. When in $S_{\ell}$ there are $r_{\ell} < s_{\ell}$ voters with money, then we proceed as follows. We add $s_{\ell} - r_{\ell}$ new voters and make these new voters and those which were left from $S_{\ell}$ approve some $\frac{s_{\ell}}{L}$ new candidates. Equal Shares will select these candidates and will eventually leave all voters from $S_{\ell}$ without money. Thus, the inductive step can be repeated until no voter from $S = S_1 \cup \ldots \cup S_x$ has money left. 

Now, observe that each voter from $S_\ell$ got exactly $\frac{s_{\ell}}{L}$ representatives. At the same time, the voters from $S$ could afford to buy candidates from $R$, where each voter from $S_\ell$ would get exactly $x^\ell$ candidates. Thus, the committee returned by Equal Shares approximates the core property by no better than:
\begin{align*}
\frac{x^\ell}{\frac{s_{\ell}}{L}} = \frac{x^\ell}{\frac{x^{\ell} - 1}{x - 1}} = \frac{x^{\ell+1} - x^\ell}{x^{\ell} - 1} \geq x-1 \text{.}
\end{align*}

Now it remains to show that in our construction $x$ is in the order of $\log(k)$. For that we first assess the number of voters $n$. Observe that during our interactive procedure, for each $\ell < x$ we have added a set of voters that is smaller than $S_\ell$. To $S_x$ we added $s_{x} - |S_i|$ voters. Thus, in total $n$ can be upper bounded by:
\begin{align*}
n < 2 |S| + s_{x} = 2L x^x + \frac{L(x^{x} - 1)}{x - 1}
\end{align*}
Consequently, since $L = \frac{n}{k}$ we have that:
\begin{align*}
k < 2 x^x + \frac{(x^{x} - 1)}{x - 1} < 3 x^x = 3\cdot 2^{x\log(x)} \text{.}
\end{align*}
For sufficiently large $x$ we have that $k < 2^{x^{1/c}}$, which implies that $x^{1/c} > \log(k)$, and so that $x > \log^c(k)$. This completes the proof.
\end{proof}

\begin{proposition}
\label{prop:x-not-core-priceable}
The Method of Equal Shares violates the core subject to $\mathcal{P}_{\mathrm{price}}$, where  $((P,k), W') \in \mathcal{P}_{\mathrm{price}}$ if $W'$ is priceable.
\end{proposition}
\begin{proof}
Consider the following profile. There are 160 voters divided into 3 groups, $V_1$, $V_2$ and $V_3$, with $|V_1| = |V_2| = 56$, and $|V_3| = 48$. There are two groups of candidate, $C_1$ and $C_2$, with $|C_1| = 20$, and $|C_2| = 16$. We set $k = 20$, thus $\nicefrac{n}{k} = 8$. 

First, let us describe how we construct voters' preferences over the first group of candidates $C_1 = \{c_1, \ldots, c_{20}\}$.  Voters from $V_1$ approve $c_1, \ldots, c_7$; voters from $V_2$ approve $c_8, \ldots, c_{14}$; finally we divide $V_3$ into 6 equal-size groups (the size of each such group is 8), and we let each such group approve one of the remaining 6 candidates from $C_1$. Intuitively, $C_1$ will be the outcome of Equal Shares for the constructed preference profile. 

Next, let us describe the preferences over the second group of candidates $C_2$. We start by dividing $C_2$ into two equal-size subgroups, $C_{2a}$ and $C_{2b}$, with $|C_{2a}| = |C_{2b}| = 8$. From $V_1$ we pick a subset $V_1'$ of 40 voters and let the voters from $V_1'$ approve all candidates from $C_{2a}$. Analogously, from $V_2$ we select a  subset $V_2'$ of 40 voters and let these approve all candidates from $C_{2b}$. Finally, we divide the voters from $V_3$ into 8 equal-size groups (each of size 6) and let each group approve one candidate from $C_{2a}$ and one candidate from $C_{2b}$ (different candidates for different groups). This completes the construction. 

Indeed, it is easy to verify that for this profile Equal Shares will pick $C_1$. On the other hand, $C_2$ is a priceable deviation for the voters from $V' = V1' \cup V_2' \cup V_3$. To see that, first note that $|V'| = 128$, and that $|V'|/|V| \cdot k = 16 = |C_2|$. Further, in $C_2$ each voter from $V_1'$ and from $V_2'$ would have 8 representatives (in $C_1$, they would have only 7). Each voter from $V_3$ will have 2 representatives in $C_2$ (while in $C_1$ they would have only one). Finally, it is easy to verify that $C_2$ is priceable with the price equal $8$. For example, for each candidate from $C_{2a}$  each voter from $V_1'$ pays $\nicefrac{1}{8}$ and the voters from the respective subgroup of $V_3$ pay $\nicefrac{1}{2}$. Thus, each candidate gets:
\begin{align*}
\frac{1}{8} \cdot 40 + 6 \cdot \frac{1}{2} = 5 + 3 = 8 \text{.}
\end{align*}
This completes the proof.
\end{proof}
\clearpage

\end{document}